\documentclass[11pt]{article}
\usepackage{graphicx} 
\usepackage[dvipsnames]{xcolor}
\usepackage{tikz, pgfplots}
\pgfplotsset{compat=1.18}
\usetikzlibrary{positioning}
\usetikzlibrary{arrows.meta}
\usetikzlibrary{patterns}
\usetikzlibrary{intersections}
\usepackage[letterpaper, margin=1in]{geometry}
\usepackage{amsmath,amssymb, amsthm}
\usepackage{authblk}

\makeatletter
\def\thm@space@setup{%
  \thm@preskip=6pt plus 2pt minus 2pt
  \thm@postskip=6pt plus 2pt minus 2pt
}
\makeatother

\providecommand{\keywords}[1]
{
  \small	
  \textbf{Keywords and phrases} #1
}

\usepackage{float}   
\usepackage{algorithm}
\usepackage{algpseudocode}
\usepackage{enumitem}
\usepackage{hyperref, cleveref}
\usepackage[breakable, skins]{tcolorbox}
\usepackage{etoolbox}
\usepackage{appendix}
\AtBeginEnvironment{appendices}{\crefalias{section}{appendix}}

\algrenewcommand\algorithmicrequire{\textbf{Input:}}
\newcommand{\MOV}{\mathsf{MOV}}
\newcommand{\A}{\mathcal{A}}
\newcommand{\CC}{\mathcal{C}}
\newcommand{\DD}{\mathcal{D}}

\newcommand{\SM}{\mathcal{S}}

\newcommand{\WP}{{\sc{Winner Prediction}}}

\newtheorem{definition}{Definition}
\newtheorem{lemma}{Lemma}
\newtheorem{claim}{Claim}
\newtheorem{fact}{Fact}
\newcommand{\LongState}[1]{\State \parbox[t]{0.88\linewidth}{#1}}
\newtheorem{theorem}{Theorem}

\renewcommand{\epsilon}{\varepsilon}
\newcommand{\veehat}[1]{\mathop{\overset{\vee}{#1}}\nolimits}
\newcommand{\epstl}{\widetilde{\epsilon}}
\newcommand{\alphatl}{\widetilde{\alpha}}

\newif\ifdraft
\drafttrue

\ifdraft
\usepackage[colorinlistoftodos,prependcaption,textsize=tiny]{todonotes}
\usepackage {xargs}
\newcommand{\sgr}[1]{\todo[color=blue!100!black!2^8]{SS: #1}}
\newcommand{\kdr}[1]{\todo[color=red,size=\tiny]{KD: #1}}
\newcommand{\dkr}[1]{\todo[color=green,size=\tiny]{DK: #1}}

\else
    \newcommand{\sgr}[1]{}
    \newcommand{\kdr}[1]{}
    \newcommand{\dkr}[1]{}

\fi

\usepackage{titlesec}
\usepackage{enumitem}

\titlespacing*{\section}{0pt}{1.2ex plus 0.4ex minus 0.2ex}{0.8ex plus 0.2ex}
\titlespacing*{\subsection}{0pt}{1.0ex plus 0.3ex minus 0.2ex}{0.6ex plus 0.2ex}

\titlespacing*{\paragraph}
{0pt}{0.5ex plus 0.2ex minus 0.1ex}{0.6em}

\setlist{nosep}

\title{Query-efficient winner prediction in district-based elections}

\author[1]{Koustav De\thanks{koustavde7@kgpian.iitkgp.ac.in}}
\author[2]{Debajyoti Kar\thanks{debajyotikar@iisc.ac.in}}
\author[3]{Swagato Sanyal\thanks{swagato.sanyal@sheffield.ac.uk}}
\affil[1]{Department of Computer Science and Engineering, Indian Institute of Technology Kharagpur, India}
\affil[2]{Department of Computer Science and Automation, Indian Institute of Science Bengaluru, India}
\affil[3]{School of Computer Science, University of Sheffield, United Kingdom}

\begin{document}
\pagenumbering{gobble}
\maketitle
\begin{abstract}
In a district-based election, $N$ voters are partitioned into $k$ districts, and each voter votes for one of $m$ candidates. Each district elects a winner using the plurality rule (i.e. the candidate getting the largest number of votes is declared the winner, breaking ties as per some fixed rule),
and the overall winner is determined by applying plurality to the district winners; we assume that there is a unique winner amongst the district winners. The margin of victory of such an election is the minimum number of votes that must be altered so that the current winner ceases to be the unique district winner. We study the problem of predicting the winner of a district-based election in the query complexity model, where one has query access to individual votes. The objective is to minimise the number of queries. This setting captures exit polling, where queries correspond to interviewing voters, and is closely related to problems in query complexity and property testing.

Assuming that the margin of victory of the election is at least $\epsilon N$, Dey, Kar and Sanyal (AAMAS 2023) gave algorithms for the case of two candidates with error probability $\delta$ and query complexity $\widetilde{O}(\frac{1}{\epsilon^6}\log^2 \frac{1}{\delta})$, which improves to $\widetilde{O}(\frac{1}{\epsilon^4}\log^2 \frac{1}{\delta})$ under the additional assumption that district populations are balanced. Our main result is an adaptive randomised algorithm that, for an arbitrary district-based election and any error parameter $\delta \in (0, 1)$, with probability at least $1-\delta$, predicts the winner correctly using $\widetilde{O}(\frac{1}{\epsilon^2} \log \frac{m}{\delta}\log \frac{1}{\delta})$ queries. In particular, we quantitatively improve the bounds of Dey, Kar and Sanyal for arbitrary district populations and extend their results to any number of candidates. Furthermore, for the special case of a constant number of candidates, our algorithm nearly matches a lower bound of $\Omega(\frac{1}{\epsilon^2}\log \frac{1}{\delta})$ on the query complexity that holds even for two candidates and a single district.
\end{abstract}

\keywords{Sampling algorithms, Query algorithms, Sub-linear algorithms, District-based Election, Winner Prediction, Margin of Victory}

\newpage
\pagenumbering{arabic}
\section{Introduction}
In a district-based election $E$, there are $N$ voters partitioned into $k$ districts. Each voter votes for some candidate in a set of $m$ candidates denoted by $\CC$. The winner of each district is defined to be the candidate getting the largest number of votes in that district. The winner of $E$  is defined to be the candidate who wins in the largest number of districts. All ties are broken in an arbitrary but fixed way.

In this paper, we are interested in predicting the winner of a district-based election in the \emph{query model}, where an algorithm can query any voter and know their vote. The queries may be made adaptively (i.e.~the decision which voter to query next may depend on the previous query outcomes) and using randomness. The objective is to minimise the number of queries made while successfully predicting the winner with probability at least $1-\delta$, for a given error parameter $\delta \in (0, 1)$. We assume that in the given election the winner is unambiguous, i.e., the winner wins in strictly more districts than any other candidate. We assume that the parameters $N, m, k$ as well as the populations of the districts are known. We call this problem $\delta$-\WP.

The \emph{margin of victory} of an election $E$, denoted by $\MOV(E)$, is the minimum number of votes that need to be altered in $E$ so that the current winner ceases to win in strictly more districts than at least one other candidate, i.e., ceases to be the unambiguous district winner. Note that it is possible for the current winner to continue being the winner in the altered election depending on the tie-breaking rule. The margin of victory \cite{BD21} is a measure of robustness of an election. We would like to bound the query complexity of our predictor in terms of $\MOV(E)$. Intuitively speaking, winner-prediction is more difficult for close elections that correspond to low margins of victory. On the other hand, for elections with high margins of victory, one can hope that few well-chosen samples from the population may potentially contain enough statistical information to reliably predict its winner. We note that we do not assume that the predictor \emph{knows} $\MOV(E)$.
\subsection{Motivation.}
$\delta$-\WP\ is motivated by a range of applications.
\subsubsection{Social choice and voting theory.} 
The setup described above abstracts a common feature of district-based electoral systems, such as those used in countries including the US, the UK, and India: voters are grouped geographically, local outcomes are determined within districts or constituencies, and these local outcomes play a central role in determining the overall political outcome. Our model is not intended to capture all institutional details of these systems, but rather to study the algorithmic problem of winner prediction under this hierarchical aggregation structure.

 The query model can be viewed as an idealised abstraction of exit polling, in which each query corresponds to interviewing a randomly chosen voter from a chosen district. Our algorithm is adaptive; in each adaptive round we check whether the query outcomes so far are statistically sufficient to predict the winner with high confidence, or we need to make more queries. However, the number of adaptive rounds is only logarithmic in the query complexity of our algorithm. Besides, in order for us to bound the query complexity in terms of $\MOV(E)$ which is unknown to the algorithm, adaptivity is inevitable.

The model also captures a broader class of hierarchical consensus mechanisms in which individual agents form local decisions that are subsequently aggregated into a global outcome. It ties naturally to the general election-prediction problem for different voting rules which, besides elections, is relevant for applications to post-election audits and social surveys (see \cite{BD21}, \cite{KDS23} and the references therein).

More broadly, our work has the potential to motivate the study of query-efficient winner prediction for district-based elections under standard voting rules other than plurality, such as approval voting, Borda count, Copeland's rule, maximin, or the Schulze rule.

Next, we discuss connections of this problem with two fields within computational complexity theory: query complexity and property testing.
\subsubsection{Query complexity.} A randomised query algorithm computing a Boolean function $f:\{0, 1\}^n \to \{0, 1\}$ on $n$ bits with error probability $\delta$ is given \emph{query access} to the input bits of $f$, and for every input bit string $x \in \{0, 1\}^n$ outputs $f(x)$ with probability at least $1-\delta$ when run on $x$. The aim of the field of query complexity is to study the power and limitation of query algorithms, to relate various query complexity measures with one another and to determine the query complexity of specific Boolean functions \cite{BDW02}. A class of Boolean functions that plays a central role in query complexity is the composition of an \emph{outer} Boolean function $f:\{0, 1\}^n \to \{0,1\}$ and an \emph{inner} Boolean function $g:\{0, 1\}^\ell \to \{0,1\}$, denoted by $f \circ g:\{0, 1\}^{n\ell} \to \{0, 1\}$, defined as: $f \circ g(x^{(1)}, \ldots, x^{(n)})=f(g(x^{(1)}), \ldots, g(x^{(n)}))$. The query complexity of composed functions has seen a long line of research \cite{AGJKLMSS18, BDK18, GLSS18, BDB20, BDBGM22, CKMPS23, S24}, and the study of query complexity of specific composed functions such as NAND tree \cite{SW86, S91}, recursive majority \cite{MNSSTX16,GJ16, ADBD26} and Tribes \cite{JK10} has contributed immensely to the development of the field.

Consider the $\delta$-\WP\ problem in which i) there are only two candidates $0$ and $1$, ii) there is an odd number $k$ of districts, and iii) the population of each district is the same odd number, namely $N/k$. In this special case, the problem is exactly designing a randomised query algorithm with error probability $\delta$ for the \emph{composed Boolean function} $\mathsf{Maj}_{k} \circ \mathsf{Maj}_{N/k}$, where for an odd positive integer $t$, $\mathsf{Maj}_t$ denotes the majority function on $t$ bits. It follows from existing lower bounds on the query complexity of composed functions \cite{AKK16, GJPW18} that any randomised query algorithm computing $\mathsf{Maj}_{k} \circ \mathsf{Maj}_{N/k}$ with error probability $1/3$ makes $\Omega(N)$ queries in the worst case. \cite{KDS23} shows that it is possible to sidestep this lower bound with high probability when the margin of victory of the given election is large. This paper improves the query complexity. This paper further extends the upper bound first to allow for arbitrary district populations (which is a generalisation of function composition where multiple inner functions are used, which are in this case majority functions on varying number of bits), and to the non-Boolean setting of arbitrary number of candidates.
\subsubsection{Property testing.} In property testing \cite{G98, G17} we study the problem of deciding whether a large object (a Boolean string, a Boolean function or a graph) satisfies some property, given query access to the object. A property is identified with the subset of objects satisfying the property. For example, the property of monotonicity of Boolean functions is defined to be the set of all monotone Boolean functions. It is assumed that the given object is either in the property or \emph{far from being} in the property. For a \emph{distance parameter} $\eta$, a Boolean function/graph is called $\eta$-far from a property if the minimum fraction of entries in the truth table/adjacency matrix that need to be altered to make it belong to the property is at least $\eta$. The objective we seek to optimise is the number of queries made to the object. Specifically, we wish to bound the query complexity in terms of the distance parameter. A property is said to be \emph{testable} if there exists a tester which can distinguish between the cases of the input object being in the property and the input object being $\eta$-far from the property, by making number of queries dependent only on the distance parameter $\eta$ (and independent of the size of the object) and with error probability $1/3$.

Once again consider the $\delta$-\WP\ where there are only two candidates, $0$ and $1$, and an odd number of districts all with odd populations (so that there is never a tie). Consider the property of $0$ being the winner. Our objective is to decide whether a given election is in the property. Note that the distance of an election $E$ where $1$ wins from this property is exactly $\MOV(E)$. However, there is one distinction between our formulation of the winner-prediction problem and the traditional property-testing setting: in our analysis we assume that even if the given election is in the property (i.e., the winner is $0$), the margin of victory is still large. In other words, we assume that the given object (election) is far from the boundary of the property regardless of whether or not it is in the property. If we suitably extend the notion of testability in this setting, our result implies that the property of $0$ being the winner is testable.\footnote{also shown by \cite{KDS23} for two candidates.} Furthermore, we extend our tester to handle multiple candidates, and maintain testability for any constant number of candidates.
\subsection{Previous work}
\label{sec:previous-work}
The problem $\delta$-\WP\ was introduced by Dey, Kar and Sanyal~\cite{KDS23}, who proposed predictors with query complexity bounded in terms of $\MOV(E)$. Dey and Mitra \cite{DM24} studied the reliability of a given survey result and their consistency with the actual election outcome. Mitra \cite{M24} studied agent-based simulation of district-based elections, and applied their model to reproduce results of elections held in India and the US.

Prior to the introduction of this problem by Dey, Kar and Sanyal, Bhattacharyya and Dey \cite{BD21} undertook a theoretical study of the task of predicting winners and estimating margins of victory of elections for various voting rules. In particular they designed a winner-predictor for the special case of a single-district election in which the winner is determined by an application of the \emph{plurality} voting rule (i.e., the candidate getting the largest number of votes is the winner, breaking ties suitably).\footnote{This is a special case of a district-based election where there is only one district.} The margin of victory for this voting rule is, up to a constant factor, the difference between the numbers of votes obtained by the winner and the runner-up. Bhattacharyya and Dey proved the following result.
\begin{lemma}[\cite{BD21}]
\label{lem:dbmain}
There exists a non-adaptive query algorithm that, for any given $\gamma,\delta \in (0, 1)$ and query access to a district-based election $E$, makes $q(\gamma, \delta) = O(\frac{1}{\gamma^2}\log \frac{1}{\delta})$ queries and outputs $m$ estimates $\{\hat{\theta}_C\}_{C \in \CC}$ of the true fractions of votes $\{\theta_C\}_{C \in \CC}$ received by the $m$ candidates satisfying:
\[\forall C \in \CC, |\theta_C-\hat{\theta}_C| \leq \gamma.\]
with probability at least $1-\delta$.
\end{lemma}
Let the margin of victory of the single-district election be at least $\epsilon N$. One can predict its winner with probability at least $1-\delta$ by simply invoking Lemma~\ref{lem:dbmain} with parameters $\gamma<\epsilon/2$ and $\delta$, and returning the candidate with the largest predicted number of votes; the algorithm uses $O(\frac{1}{\epsilon^2}\log\frac{1}{\delta})$ queries. Notably, the query complexity of the algorithm is independent of the population size $N$ and the number of candidates $m$. 

It is not difficult to show that this algorithm can be adapted so that the knowledge of $\epsilon$ is not needed, and yet with probability at least $1-\delta$ it makes $\widetilde{O}(\frac{1}{\epsilon^2}\log\frac{1}{\delta})$ queries\footnote{Throughout the paper, the notation $\widetilde{O}$ and $\widetilde{\Omega}$ hide factors polylogarithmic in $1/\epsilon$ and $\log 1/\delta$.} and returns the correct winner. The idea is to iteratively use Lemma~\ref{lem:dbmain} to compute the fractional vote share estimates for various settings of $\gamma$ starting from a large value of $\gamma$ (say $\gamma=1/2$) and decreasing it by a factor of $2$ in each step, until the difference between the two leading estimates is large compared to $\gamma$. In order for all the estimates to be correct with high probability, the $\delta$ parameter for each invocation will have to be set lower, thereby tolerating a blow-up in the query complexity by a factor of $\log (1/\epsilon)$. An analysis of this algorithm shows that comparing the difference of the two leading estimates with the parameter $\gamma$ reveals whether the estimates can be relied on for the winner-prediction or whether more queries are necessary.

Bhattacharyya and Dey also observed that a matching lower bound of $\Omega(\frac{1}{\epsilon^2}\log \frac{1}{\delta})$
for the query complexity of this problem follows from standard facts in statistics for $\epsilon = \Omega(1/\sqrt{N})$, even when there are only two candidates. As a single-district election is a special case of a district-based election where there is only one district, the above lower bound holds for $\delta$-\WP.

Dey, Kar and Sanyal proved the following results.
\begin{theorem}[Dey, Kar and Sanyal \cite{KDS23}]
\label{thm:dks}
Consider the $\delta$-\WP\ problem when there are two candidates. Let the input election $E$ satisfy $\MOV(E)\geq \epsilon N$.
\begin{enumerate}[label=(\roman*)]
\item Assume that the population of each district is within a constant factor of the average population $N/k$. Then, there exists an algorithm that, with probability at least $1-\delta$, makes at most $\widetilde{O}(\frac{1}{\epsilon^4}\log^2\frac{1}{ \delta})$ queries and outputs the true winner.
\item Assume that the population of each district is arbitrary. Then, there exists an algorithm that, with probability at least $1-\delta$, makes at most $\widetilde{O}(\frac{1}{\epsilon^6}\log^2\frac{1}{ \delta})$ queries and outputs the true winner.
\end{enumerate}
\end{theorem}
The authors also proved a matching lower bound for a natural class of algorithms; see Section~\ref{sec:dks-ideas} for a discussion.
\subsection{Our result}
In this paper, we quantitatively improve the result of Dey, Kar and Sanyal and refine it by relaxing several assumptions. In particular, we propose an algorithm with optimal query complexity (up to logarithmic factors) for arbitrary district populations when the number of candidates $m$ is a constant.
\begin{theorem}
\label{thm:mainthm}
Consider the general $\delta$-\WP\ problem with arbitrary district populations and arbitrary number of candidates. There exists an algorithm $\A$ with the following guarantee: for every $\epsilon, \delta \in (0, 1) $ if the input election $E$ satisfies $\MOV(E) \geq \epsilon N$ then with probability at least $1-\delta$, $\A$ makes at most $\widetilde{O}\left(\frac{1}{\epsilon^2}\log \frac{m}{\delta}\log\frac{1}{\delta}\right)$ queries on $E$ and outputs the true winner of $E$.
\end{theorem}
In particular for the special case of a constant number of candidates, our algorithm has  query complexity $\widetilde{O}(\frac{1}{\epsilon^2}\log^2\frac{1}{\delta})$, which is optimal up to factors polylogarithmic in $1/\epsilon$ and $\log 1/\delta$.
\section{Proof ideas}
\label{sec:proof-ideas}
In this section we discuss the ideas and techniques underlying our proof. We first discuss the high-level ideas of the proof of Theorem~\ref{thm:dks}
 by Dey, Kar and Sanyal. We then discuss the barriers that their approach faces in reducing the query complexity as well as relaxing the assumptions (two candidates, equal district populations). Finally we discuss the new technical ingredients that we develop in this work to bypass these obstacles and prove our main theorem. Throughout this section, we assume that $\epsilon$ is known. This assumption is removed by iteratively guessing and affirming the guess, along the lines of the ideas outlined in Section~\ref{sec:previous-work}.
 
 For simplicity, throughout this section we assume that $\delta$ is a constant (e.g. $1/3$). A crucial ingredient in all the proofs is the single-district winner predictor by Bhattacharyya and Dey. It is a randomised algorithm and errs with some probability. However, in this overview we assume that it always outputs correctly as long as the number of votes queried is large enough with respect to the margin of victory. In the actual proof, however, we need to set the $\delta$ parameter of this algorithm low enough so that all its invocations can be relied on. This leads to a logarithmic blow-up in its query complexity, which we hide in the $\widetilde{O}$ notation. We also assume for simplicity in this section that $\epsilon=\Omega(\frac{1}{k})$.

We note here that arguably the most significant technical challenge that we deal with in this work is to handle arbitrary number of candidates. The mathematical structure of elections with more than $2$ candidates (even $3$ candidates) and high MOV turns out to be significantly deeper and more subtle than elections with only two candidates and bounded MOV. We outline the techniques that we develop and employ to handle arbitrary number of candidates in Section~\ref{intro:m>2}. Until before Section~\ref{intro:m>2} we discuss other challenges that we navigate and to keep the discussion there focussed we assume that we have only $2$ candidates.
 \subsection{Proof ideas of Theorem~\ref{thm:dks}}
\label{sec:dks-ideas}
Recall that in Theorem~\ref{thm:dks} there are only two candidates.
 Let $W_1, \ldots, W_k$ be the district winners. A key technical lemma (Lemma~\ref{lem:district-lebel-lb}) that they showed asserts that the margin of victory of the single-district election $E'$ defined by $W_1, \ldots, W_k$ is $\Omega(\epsilon k)$. We note that Lemma~\ref{lem:district-lebel-lb} also holds for an arbitrary number of candidates.

 The winner of $E$ is the same as the winner of $E'$, and $\MOV(E')=\Omega(\epsilon k)$. Hence, if we had query access to the $W_i$s, we could have sampled $\widetilde{O}(1/\epsilon^2)$ of them and used the predictor by Bhattacharyya and Dey to output a prediction. Since we do not have direct access to the $W_i$s, we instead predict $W_i$ by running the predictor by Bhattacharyya and Dey on the single-district election defined by the $i$-th district, which we denote by $E_i$.

 But the single-district predictor needs a lower bound on the $\MOV$ to output reliably. What guarantee do we have on $\MOV(E_i)$? The $E_i$s with $\MOV(E_i) \ll \epsilon\frac{N}{k}$ can be shown to not matter much, i.e., it does not hurt the reliability of the final prediction to a significant degree if the predicted winner of $E_i$ is wrong. This is a consequence of the lower bound of $\epsilon N$ on the margin of victory of the overall election $E$. So, while estimating the winner of $E_i$, we assume that $\MOV(E_i)=\Omega(\epsilon\frac{N}{k})$.  However, the query complexity of the single-district predictor within the $i$-th district depends on the $\MOV(E_i)$ \emph{viewed as a fraction of the district population}.
 \subsubsection{Part 1.} Here, we additionally assume that the district populations are within a constant factor of $N/k$. In this case, the bound $\MOV(E_i)=\Omega(\epsilon\frac{N}{k})$ implies that $\MOV(E_i)$ is at least $\Omega(\epsilon)$ fraction of the district population.  Thus, the single-district predictor samples $\widetilde{O}(1/\epsilon^2)$ within the $i$-th district to predict $W_i$. In this case the final predictor takes the following form:
 \begin{enumerate}
\item Sample $\ell_1:=\widetilde{O}(1/\epsilon^2)$ districts. Let the corresponding single-district elections be $E_{i_1}, \ldots, E_{i_{\ell_1}}$.
\item For each $j \in \{1, \ldots, \ell_1\}$, predict the winner of $E_{i_j}$ by using the single-district predictor by Bhattacharyya and Dey assuming that $\MOV(E_{i_j})=\Omega(\epsilon\frac{N}{k})$. Let the predicted winner be $\widetilde{W}_j$. This step involves sampling and querying $\ell_2:=\widetilde{O}(1/\epsilon^2)$ votes from each sampled district.
\item Output the winner of the single-district election defined by $\widetilde{W}_1, \ldots, \widetilde{W}_{\ell_1}$.
 \end{enumerate}
 The overall query complexity is $\ell_1 \ell_2=\widetilde{O}(1/\epsilon^4)$.
 \subsubsection{Part 2.} However, when the district populations are arbitrary, a lower bound $\MOV(E_i)=\Omega(\epsilon\frac{N}{k})$ does not imply that $\MOV(E_i)$ is at least $\Omega(\epsilon)$ fraction of the population of the $i$-th district (denoted by $n_i$). Nevertheless the authors could show that the number of districts for which $\MOV(E_i)=\Omega(\epsilon^2 n_i)$ is at least $k(\frac{1}{2}+\Omega(\epsilon))$. The intuition behind this result is that if $\MOV(E_i)=\Omega(\epsilon\frac{N}{k})$ and $\MOV(E_i)\ll\epsilon^2 n_i$, then $n_i \gg N/\epsilon k$, i.e., the population of the $i$-th district is far larger than the average population of a district. Thus, the number of such districts is $\ll\epsilon k$.

 In this case we need to invoke Bhattacharyya and Dey's algorithm for an $\MOV$ lower bound of $\Omega(\epsilon^2 n_i)$; so $\widetilde{O}(\frac{1}{(\epsilon^2)^2}) = \widetilde{O}(\frac{1}{\epsilon^4})$ votes are queried from each sampled district. The query complexity of the overall predictor is thus $\widetilde{O}(\frac{1}{\epsilon^2}\cdot \frac{1}{\epsilon^4})= \widetilde{O}(\frac{1}{\epsilon^6})$.
 \subsection{Reducing the query complexity}
 \label{intro:red_compl}
 Now we outline our key technical ingredient that enables us to reduce the query complexity to a quadratic dependence on $1/\epsilon$ (up to logarithmic factors). To keep the discussion focused we assume that there are only two candidates $A$ and $B$, and the population of each district is within a constant factor of $N/k$. Removing these assumptions will need additional technical ideas and will be discussed in subsequent sections. 

 How tight is the bound in Theorem~\ref{thm:dks}(i)? In particular, can we reduce either $\ell_1$ or $\ell_2$ in the previous section? Dey, Kar and Sanyal proved that neither $\ell_1$ nor $\ell_2$ can be reduced while still preserving the performance guarantee of the algorithm, leading to a lower bound of $\widetilde{\Omega}(1/\epsilon^4)$ for the class of algorithms that fit the template of the algorithm presented in the previous section with the freedom to set $\ell_1$ and $\ell_2$. An analysis of elections $E_\beta$  discussed next also recovers this lower bound.

Consider the following election $E_\beta$, parameterised by a parameter $\beta \in (0, 1)$.
\begin{tcolorbox}
    [breakable, width=\linewidth, colback=white!95!blue, title = $E_\beta$]
    The winner wins in $k(\frac{1}{2}+\epsilon^\beta)$ districts. In each of those districts, the $\MOV$ of the corresponding single-district election is $\epsilon^{1-\beta}\cdot\frac{N}{k}$. In each of the remaining districts in which the other candidate wins, we don't specify the margin of victory; in particular we allow the $\MOV$ to be as small as $1$.
\end{tcolorbox}
$\MOV(E_\beta)$ can be easily verified to be $\epsilon N$. The $\MOV$ of the single-district election defined by the district winners is $\epsilon^\beta k$; so the algorithm in Section~$\ref{sec:dks-ideas}$ must set $\ell_1$ to $\widetilde{\Theta}(\frac{1}{\epsilon^{2\beta}})$. In particular if $\beta \approx 1$, $\ell_1 = \widetilde{\Theta}(\frac{1}{\epsilon^2})$.

Now, in order to reliably predict the winner for the districts where the overall election winner wins, $\ell_2$ must be set to $\widetilde{\Theta}(1/\epsilon^{2(1-\beta)})$. Note that if $\beta \approx 0$, $\ell_2$ needs to be set to $\widetilde{\Theta}(\frac{1}{\epsilon^2})$.

With the aforementioned settings of $\ell_1$ and $\ell_2$, irrespective of the accuracy of the predictions for the districts in which the true winner does not win, we will have enough statistical information to correctly predict the winner.

As discussed above, the settings of $\ell_1$ and $\ell_2$ in the algorithm in Section~\ref{sec:dks-ideas} cannot be reduced in general. However, for any $\beta$, if we are allowed to set $\ell_1$ and $\ell_2$ depending on $\beta$, the overall query complexity is at most $\widetilde{\Theta}(\frac{1}{\epsilon^{2\beta}}\cdot \frac{1}{\epsilon^{2(1-\beta)}})=\widetilde{\Theta}(\frac{1}{\epsilon^2})$.

$E_\beta$ demonstrates a trade-off between the $\MOV$ of the election defined by the $W_i$s, and that of the elections defined by the districts that matter for the prediction task, which can potentially be exploited to reduce the query complexity. The following lemma establishes such a trade-off for any district-level election with $\MOV$ at least $\epsilon N$.
\begin{lemma}
\label{lem:simple-alpha}
Let $E$ be a district-level election with two candidates satisfying $\MOV(E) \geq \epsilon N$ and $4/k\leq \epsilon\leq 1/2$. Then there exists an $\alpha \in [\frac{\epsilon}{8}, \frac{1}{2}]$ such that there are at least $k(\frac{1}{2}+ \alpha)$ districts each with $\MOV$ at least $\widetilde{\Omega}(\frac{\epsilon N}{\alpha k})$ in which the overall winner wins.
\end{lemma}

Lemma~\ref{lem:simple-alpha} is a simplified version of a more general lemma (Lemma~\ref{lem:mov-lemma}) which is used to prove Theorem~\ref{thm:mainthm}.  Lemma~\ref{lem:mov-lemma} handles general district-based elections with arbitrary number of candidates, is integrated with structural properties of such elections and is more technical. Lemma~\ref{lem:mov-lemma} is discussed in Section~\ref{intro:m>2}. For a proof of Lemma~\ref{lem:simple-alpha} see \Cref{sec:simple-alpha-proof}.

Lemma~\ref{lem:simple-alpha} enables a cheaper algorithm as follows. Let $t:=\frac{\epsilon}{\alpha}$, where $\alpha$ is given by Lemma~\ref{lem:simple-alpha}. First, assume that the $\alpha$
 is known to the algorithm. In that case, we set $\ell_1$ to $\widetilde{\Theta}(1/\alpha^2)$, $\ell_2$ to $\widetilde{\Theta}(1/t^2)$ and run the algorithm in Section~\ref{sec:dks-ideas}.  We've set $\ell_2$ large enough to correctly predict the winners of each of the districts mentioned in Lemma~\ref{lem:simple-alpha}. Since their fraction is at least $\frac{1}{2}+\alpha$, in our sampled election more than half of the predicted winners would be the actual overall winner, and hence the output will be correct. The overall query complexity is $\widetilde{O}(\frac{1}{\alpha^2}\cdot \frac{1}{t^2})=\widetilde{O}(\frac{1}{\epsilon^2})$.

But the algorithm does not know $\alpha$. To deal with that, it breaks up the range $[\epsilon/8, 1/2]$ into $O(\log \frac{1}{\epsilon})$ intervals $[\epsilon/8, \epsilon/4),[\epsilon/4, \epsilon/2), \ldots, [1/4, 1/2]$ (assuming $\epsilon$ to be an inverse power of $2$).
For each of these intervals, it guesses $\alpha$ to be in that interval and runs the aforementioned algorithm. This adds an additional factor of $O(\log(1/\epsilon))$ to the query complexity.

A correct guess would lead to the correct prediction. But can a wrong guess lead to a wrong prediction? The key insight here is that the Bhattacharyya-Dey single-district winner predictor can also detect whether the number of queries into a district is insufficient to reliably predict its winner, or whether the number of confident predictions of district winners is too small to predict the overall winner reliably based on them (see the discussion in Section~\ref{sec:previous-work}).  So the algorithm can be made to discard guesses of $\alpha$ if the query outcomes do not present clear statistical evidence of a candidate winning more than half of the districts.
\subsection{Handling arbitrary district populations}
The improvement outlined in the preceding section hinges on the following observation: an $\MOV$ of $\widetilde{\Omega}(t\frac{N}{k})$ for a district is indeed $t$-fraction of the district population, as the population of each district is within a constant factor of $N/k$. Thus, sampling $\widetilde{O}(1/t^2)$ votes from such a district suffices to predict its winner reliably. While Lemma~\ref{lem:simple-alpha} holds for arbitrary district populations, the aforementioned argument fails when the district populations are arbitrary. 

To tackle this, the first step is to ignore districts which are too large: with population $\gg\frac{N}{\epsilon k}$ as their total number is $\ll \epsilon k$, which is too small to matter.

Next, the remaining districts are grouped into classes based on their populations; within each class (except one) all districts have their populations within a factor of $2$ of one another. We have the set of districts $\DD_0$ with populations at most $ N/k$,
and for each $i$ in $[1, \ldots, \approx \log (1/\epsilon)]$ we define a set $\DD_i$ of districts with population lying in $(\frac{2^iN}{k}, \frac{2^{i+1}N}{k}]$.\footnote{We remark that in the formal proof the notations used are slightly different from the ones used in this overview.}

Now for each $i$ we estimate the number of districts in $\DD_i$ that satisfy the conditions in Lemma~\ref{lem:simple-alpha}: bounded $\MOV$ and the winner being the same as the overall winner of the district-level election, within an additive error of $O(\frac{\alpha}{\log(1/\epsilon)}k)$. Since there are $\approx \log (1/\epsilon)$ district classes, the sum of these estimates will be an estimate of the total number of districts (leaving out the districts with very large populations as explained before) that satisfy the Lemma conditions within an additive $O(\alpha k)$.

For a specific $i$, the number of districts in $\DD_i$ is at most $O(k/2^i)$; so the tolerable error of $O(\frac{\alpha}{\log(1/\epsilon)}k)$ is at least a $\Omega(\frac{\alpha}{\log(1/\epsilon)}2^i)$ fraction of the number of such districts. So it suffices to sample at most $\widetilde{O}(\frac{\log^2(1/\epsilon)}{\alpha^22^{2i}})$ such districts. Note that here we save a factor of $2^{2i}$ compared to previous section.

However, for a district in $\DD_i$, the $\MOV$ lower bound of $\widetilde{\Omega}(t
\frac{N}{k})$ is $\widetilde{\Omega}(\frac{t}{2^i})$ fraction of its population; so we need to query $\widetilde{O}(\frac{2^{2i}}{t^2})$ votes from each sampled district in $\DD_i$, which is a factor of $2^{2i}$ worse than the previous section.

Hence, the total number of queries for each $i$ is $\widetilde{O}(\frac{\log^2(1/\epsilon)}{\alpha^22^{2i}}\cdot \frac{2^{2i}}{t^2})=\widetilde{O}(1/\epsilon^2)$. Repeating it for each population class with appropriate settings of the error parameter adds an additional multiplicative overhead of  $O(\mathsf{poly}\log(1/\epsilon))$ to the query complexity.

We would like to contrast our approach with that of Dey, Kar and Sanyal; they incurred an overhead of a factor of $\Omega(1/\epsilon^2)$ while going from part 1 to part 2 of Theorem~\ref{thm:dks} by tolerating a deterioration of the $\MOV$ guarantee of the single-district elections defined by the districts with arbitrary population. We instead incur an overhead of only $O(\mathsf{poly}\log (1/\epsilon))$ by treating each population class differently.
\subsection{Handling more than two candidates}
\label{intro:m>2}
As mentioned before, the most significant technical challenge that we faced is to extend our proof to handle more than $2$ candidates. All the approaches discussed so far crucially rely on a property of a district-based election involving $2$ candidates: to win an election it is both necessary and sufficient to win in more than half the districts. All the algorithms discussed so far focus on finding a candidate who wins in more than half the districts, and declare such a candidate the winner. However, this characterisation does not generalise to $m>2$: while it is necessary for the winner to win in at least $k/m$ districts, it is no longer a sufficient condition to win the election.

So, what condition should an algorithm look for? We look for the condition implied by the definition of winner: $W$ is the winner if and only if for every other candidate $C$, $W$ wins in more districts than $C$. So, at a high level the algorithm finds a candidate $W$ that satisfies the above condition.

Note that in Theorem~\ref{thm:mainthm} the query complexity has a logarithmic dependence on $m$, and it is a good time to outline the source of this factor. As per the strategy hinted above, pairwise comparisons between candidates are involved. This will translate to accurately estimating various statistics for all the candidates, using the Bhattacharyya-Dey algorithm. In order for all of those estimates to be reliable, the error parameter $\delta$ needs to be inversely proportional to $m$. Since the query complexity of the Bhattacharyya-Dey algorithm has a logarithmic dependence on $1/\delta$, our algorithm ends up suffering an extra $\log m$ factor in its query complexity.

For the rest of this section, we fix two candidates $W$ and $C$, and explain at a high level what kind of evidence our algorithm gathers that certifies that $W$ wins in more districts than $C$.  We also switch back to the equal district population setting to focus on the technical ideas developed to handle more than $2$ candidates; arbitrary district populations are handled by integrating these ideas with the ones discussed in the preceding section.

Informally speaking, for a parameter $\alpha \in (0, 1)$, call a district $d$ to be \emph{$(\alpha, W)$-rich} if $W$ is the winner of $d$ and no other candidate can be made the winner by changing $\widetilde{O}(\frac{\epsilon N}{\alpha k})$ votes, and call $d$ to be \emph{$(\alpha, C)$-poor} if $C$ is either the winner of $d$ or can be made the winner by changing $\widetilde{O}(\frac{\epsilon N}{\alpha k})$ votes. 
Let $R(\alpha, W)$ and $P(\alpha, C)$ stand for the numbers of $(\alpha, W)$-rich and $(\alpha, C)$-poor districts respectively.
We note that these notations are only for this overview, and are not used in the formal proof.

Observe that $R(\alpha, W)$ is a lower bound on the number of districts in which $W$ wins, and $P(\alpha, C)$ is an upper bound on the number of districts in which $C$ wins. So, if for some value of $\alpha$ we can certify that $R(\alpha, W) > P(\alpha, C)$, then that will also certify that $W$ wins in more districts than $C$.

This is enabled by Lemma~\ref{lem:mov-lemma} which, as mentioned before, is a generalisation of Lemma~\ref{lem:simple-alpha}. At a high level, Lemma~\ref{lem:mov-lemma} asserts that there exists an $\alpha=\Omega(\epsilon)$ (which may depend on $W$ and $C$) such that if $W$ is indeed the true winner of the election then $R(\alpha, W) - P(\alpha, C) = \Omega(\alpha k)$.

For simplicity, suppose that $\alpha$ is known to the algorithm; this assumption can be removed along the lines of the ideas discussed in Section~\ref{intro:red_compl}. $\widetilde{O}(\frac{\alpha^2}{\epsilon^2})$ queries into a district suffice to figure out whether it is $(\alpha, W)$-rich or $(\alpha, C)$-poor (note that the assumption that all district populations are within a constant factor of $N/k$ is used here). We can use it simultaneously for all $m$ candidates without the query complexity scaling with $m$, using the Bhattacharyya-Dey algorithm.

In order to estimate $R(\alpha, W)$ and $P(\alpha, C)$ within an additive $O(\alpha k)$, it is sufficient to sample $\widetilde{O}(\frac{1}{\alpha^2})$ districts. However this is for a fixed $W$ and $C$. If we want our estimates to be reliable for all candidates with high probability, it is sufficient for the estimate for each candidate to be accurate with error probability inversely proportional to $m$: as discussed before this blows up the number of districts sampled to $\widetilde{O}(\frac{\log m}{\alpha^2})$. In total, the query complexity sums up to $\widetilde{O}(\frac{\alpha^2}{\epsilon^2}\cdot\frac{\log m}{\alpha^2})=\widetilde{O}(\frac{\log m}{\epsilon^2})$.
 
Proving Lemma~\ref{lem:mov-lemma}  involves a deep analysis of the structure of elections with $\MOV \geq \epsilon N$. Altering votes optimally to make $C$ win in at least as many districts as $W$ wins in, potentially involves three different kinds of districts:
\begin{itemize}
    \item $\DD_1$: Districts where neither $W$ nor $C$ wins in the original election, and $C$ is made the winner by altering fewest possible votes. Doing this for one such district reduces the difference of votes between $W$ and $C$ by $1$.
    \item $\DD_2$: Districts where $W$ wins in the original election, and $C$ is made the winner by altering fewest possible votes. Doing this for one such district reduces the difference of votes between $W$ and $C$ by $2$.
    \item $\DD_3$: Districts where $W$ wins in the original election, and a candidate other than $W$ and $C$ is made the winner by altering fewest possible votes. Doing this for one such district reduces the difference of votes between $W$ and $C$ by $1$.
\end{itemize}
\Cref{lem:mov-lemma} is derived by a careful study of the properties of these three sets of districts, the optimality of the strategy under consideration and the guarantee $\MOV \geq \epsilon N$ of the election.

\section{Organisation}
For ease of exposition, we first present our result for the special case of balanced district populations. Towards this, in Section~\ref{sec:structural-result} we present structural results, some of which hold only for elections with balanced district populations. In Section~\ref{sec:alg-equal-population} we present our algorithm for balanced district populations. In Section~\ref{sec:structure-arbitrary} we refine and extend the structural results of Section~\ref{sec:structural-result} to elections with arbitrary district populations. Finally, we present the proof of Theorem~\ref{thm:mainthm} (for arbitrary district-based elections) in Section~\ref{sec:alg-arbitrary-population}.

\section{Structural results for election with $\MOV(E)\ge \epsilon N$}
\label{sec:structural-result}

In this section, we derive certain structural properties of an election $E$ satisfying $\MOV(E)\ge \epsilon N$ that we shall later leverage in designing our algorithm. We assume that $\epsilon$ is a power of $1/2$; this can be attained by rounding $\epsilon$ down if necessary, thereby reducing it by a factor of at most $2$. In Section \ref{sec:proof-ideas}, we implicitly assumed that the number of districts $k$ is sufficiently large, i.e., at least $\Omega(\frac{1}{\epsilon})$. In the complementary case when $k$ is small, a number of additional subtleties arise that need to be handled carefully.


We now introduce some notations. We call a value $x\in (0,1/2]$ to be \emph{big} if it satisfies $x>\frac{2^{15}}{k}\log\frac{\log(1/x)}{\delta x}$, and \emph{tiny} otherwise. Define a function $g\colon (0,1/2]\rightarrow\mathbb{R}$ as $g(x)=\log(1/x)$ if $x$ is big, and $g(x)=\log(2k)$ otherwise. Let $\DD$ denote the set of all districts and for any candidate $C\in \CC$, let $\DD_{C}\subseteq \DD$ be the set of districts where $C$ wins. Let $W\in \CC$ be the winner of the election $E$. We assume that we are in the setting where $|\DD_W|>|\DD_C|$ for every candidate $C\neq W$. We will need the following result from \cite{KDS23}.
\begin{lemma}[\cite{KDS23}]
\label{lem:district-lebel-lb}
Let $\MOV(E) \geq \epsilon N$. Then $|\DD_W| - |\DD_C| \geq \frac{\epsilon k}{3}$, for all $C\in \CC \setminus \{W\}$.
\end{lemma}


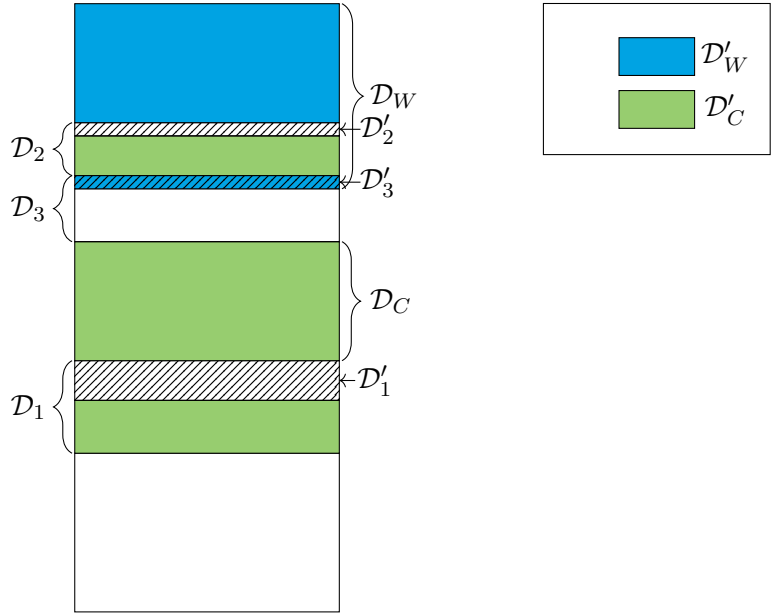
\begin{figure}[htbp]
    \centering
    \begin{tikzpicture}[scale=0.7]
        \draw (0,0) rectangle (5,3);
        \draw[fill=YellowGreen] (0,3) rectangle (5,4);
        \draw[pattern=north east lines, pattern color=black] (0,4) rectangle (5,4.75);
        \draw[fill=YellowGreen] (0,4.75) rectangle (5,7);
        \draw (0,7) rectangle (5,8);
        \draw[preaction={fill, Cerulean}, pattern=north east lines, pattern color=black] (0,8) rectangle (5,8.25);
        \draw[fill=YellowGreen] (0,8.25) rectangle (5,9);
        \draw[pattern=north east lines, pattern color=black] (0,9) rectangle (5,9.25);
        \draw[fill=Cerulean] (0,9.25) rectangle (5,11.5);

        \node[rectangle,draw=black!100,minimum width=3cm,minimum height=2cm] (Enclosing_Legends) at (11,10.075) {};
        
        \node[rectangle,draw=black!100,fill=Cerulean,minimum width=1cm,minimum height=0.5cm] (D_W_Prime_Legend) at (11,10.5) {};
        \node[rectangle,draw=black!100,fill=YellowGreen,minimum width=1cm,minimum height=0.5cm] (D_C_Prime_Legend) at (11,9.5) {};
        \node at (D_W_Prime_Legend.east) [xshift=0.4cm] {$\DD_{W}^{\prime}$};
        \node at (D_C_Prime_Legend.east) [xshift=0.4cm] {$\DD_{C}^{\prime}$};
        
        \draw [decorate, decoration={brace,amplitude=7pt,raise=1pt}]
            (0,3) -- (0,4.75)
            node [midway, anchor=east, xshift=1mm, outer sep=10pt]{$\DD_1$};

        \draw [decorate, decoration={brace,amplitude=7pt,raise=1pt}]
            (0,7) -- (0,8.25)
            node [midway, anchor=east, xshift=1mm, outer sep=10pt]{$\DD_3$};

        \draw [decorate, decoration={brace,amplitude=7pt,raise=1pt}]
            (0,8.25) -- (0,9.25)
            node [midway, anchor=east, xshift=1mm, outer sep=10pt]{$\DD_2$};

        \draw[<-] (5,4.375) -- (5.3,4.375)
            node[midway,right] {$\DD_1^{\prime}$};

        \draw [decorate, decoration={brace,mirror,amplitude=7pt,raise=1pt}]
            (5,4.75) -- (5,7)
            node [midway, anchor=west, outer sep=8pt]{$\DD_{C}$};
        
        \draw[<-] (5,8.125) -- (5.4,8.125)
            node[midway,right] {$\DD_3^{\prime}$};

        \draw[<-] (5,9.125) -- (5.4,9.125)
            node[midway,right] {$\DD_2^{\prime}$};

        \draw [decorate, decoration={brace,mirror,amplitude=8pt,raise=1pt}]
            (5,8) -- (5,11.5)
            node [midway, anchor=west, outer sep=8pt]{$\DD_{W}$};
    \end{tikzpicture}

    \caption{Division of $\DD$ into various types of districts. $\DD_1, \DD_2,\DD_3$ are the districts where votes are altered in the strategy $\SM$. $\DD'_1,\DD'_2,\DD'_3$ are the districts with large values of $h(d)$ guaranteed by Lemma \ref{lem:mov-lemma}. The blue and green districts correspond to $\DD'_W$ and $\DD'_C$, respectively.}
    \label{fig:district-decomposition}
\end{figure}

Consider any candidate $C\neq W$. Let $\Delta_C:= |\DD_W|-|\DD_C|$. Note that by our assumption $\Delta_C >0$. We call a strategy to alter a subset of votes in $E$ to be $C$-\emph{feasible} if in the resulting election obtained after altering the votes, the number of districts where $C$ is the winner is at least that of $W$. A $C$-feasible strategy is said to be \emph{optimal} if it alters the minimum number of votes. Observe that an optimal $C$-feasible strategy $\mathcal{S}$ involves altering votes in three types of districts -- a set $\DD_1\subseteq \DD\setminus (\DD_W\cup \DD_C)$ where the winner is changed to $C$ from some candidate in $\CC\setminus \{W,C\}$; a set $\DD_2\subseteq \DD_W$ where the winner is changed from $W$ to $C$; and a set $\DD_3 \subseteq \DD_W$ where the winner is changed from $W$ to some candidate in $\CC \setminus \{W,C\}$ (see Figure \ref{fig:district-decomposition}). Note that $|\DD_1|+2|\DD_2|+|\DD_3|\in [\Delta_C, \Delta_C+1]$. For $d\in \DD_1\cup \DD_2\cup \DD_3$, let $h(d)$ be the number of votes altered in district $d$ in the strategy $\SM$. Since $\MOV(E)\ge\epsilon N$, it must hold that $\sum_{d\in \DD_1\cup \DD_2\cup \DD_3} h(d) \ge \epsilon N$. We show the following lemma.

\begin{lemma}
\label{lem:mov-lemma}
    For all candidates $C \neq W$, there exists an $\alpha_C\in (0, 1/2]$ such that there are sets $\DD'_i \subseteq \DD_i$, for $i\in [3]$, with $|\DD'_1\cup \DD'_2\cup \DD'_3|\ge \max\{\alpha_C k,1\}$ and $h(d) > \frac{\epsilon N}{2^5\alpha_{C}g(\epsilon) k}$, for all $d\in \DD'_1\cup \DD'_2\cup \DD'_3$. Further $\alpha_C$ is an integral power of $1/2$, and $\alpha_C\ge \epsilon/16$, if $\epsilon$ is big, and $\alpha_C \ge 1/(2k)$, otherwise.
\end{lemma} 
\begin{proof}
    We first show the existence of a value $\beta_C \in (0,1/2]$ such that there are sets $\DD'_i \subseteq \DD_i$, for $i\in [3]$ and satisfying $|\DD'_1\cup \DD'_2\cup \DD'_3|\ge \beta_{C}k$ and $h(d)>\frac{\epsilon N}{16\beta_{C}g(\epsilon)k}$, for all $d\in \DD'_1\cup \DD'_2\cup \DD'_3$; later we shall set $\alpha_C$ to be the largest power of $1/2$ not exceeding $\beta_C$. Let $\Delta:=|\DD_1\cup \DD_2\cup \DD_3|$. We sort the districts $d$ of $\DD_1\cup \DD_2\cup \DD_3$ in non-increasing order of $h(d)$, and let $M_1\ge M_2 \ge \ldots \ge M_{\Delta}$ be the values of $h(d)$ in these districts. Since $\MOV(E)\ge\epsilon N$, we have that $\sum_{i\in [\Delta]} M_i \ge \epsilon N$.
    
    We first consider the case when $\epsilon$ is tiny. We will show that there exists an $i\in [\Delta]$ such that $M_i > \frac{\epsilon N}{16i\log (2k)}$. This would imply our desired claim by taking $\beta_C=\frac{i}{k}$ and observing that $M_1\ge M_2\ge \ldots \ge M_i$. Towards a contradiction, assume that for all $i\in [\Delta]$, it holds that $M_i \le \frac{\epsilon N}{16i\log (2k)}$. Then 
    \[ \sum_{i\in [\Delta]} M_i \le \frac{\epsilon N}{16\log (2k)}\sum_{i\in [\Delta]} \frac{1}{i}\le \frac{\epsilon N}{16},\]
    where the last inequality follows since $\Delta \le k$, and we obtain a contradiction.

    Next, let $\epsilon$ be big. In this case, we show that there exists an integer $i\in [\lceil \epsilon k/8\rceil, \Delta]$ such that $M_i > \frac{\epsilon N}{16i\log (1/\epsilon)}$, from which the claim would again follow by taking $\beta_C = \frac{i}{k}$. Assuming to the contrary that for all $i \in [\lceil \epsilon k/8\rceil, \Delta]$, $M_i \le \frac{\epsilon N}{16i\log (1/\epsilon)}$, we then have
    \[ \sum_{i=\lceil \epsilon k/8\rceil}^{\Delta} M_i \le \frac{\epsilon N}{16\log(1/\epsilon)}\sum_{i=\lceil \epsilon k/8\rceil}^{\Delta} \frac{1}{i}\le \frac{\epsilon N}{16\log(1/\epsilon)} \left(\ln\left(\frac{\Delta}{\lceil \epsilon k/8\rceil-1}\right)+1\right)\le \frac{\epsilon N}{2}.\]
    This implies that by altering at most $\epsilon N/2$ votes, we can obtain another election $E'$ where $W$ wins in at most $\epsilon k/8$ districts more than $C$. However, since $\MOV(E)\ge\epsilon N$, $W$ must still be the winner in $E'$ and $\MOV(E')\ge \epsilon N - \frac{\epsilon N}{2}=\frac{\epsilon N}{2}$. Applying Lemma \ref{lem:district-lebel-lb} to $E'$, it follows that $W$ must win in at least $\epsilon k/6$ districts more than any other candidate, contradicting the fact that $W$ only wins in at most $\epsilon k/8$ districts more than $C$.

    Thus, in either case, we obtain a value $\beta_C$ such that $|\DD'_1\cup \DD'_2\cup \DD'_3|\ge \beta_{C}k$ and $h(d)>\frac{\epsilon N}{16\beta_{C}g(\epsilon)k}$, for all $d\in \DD'_1\cup \DD'_2\cup \DD'_3$. Let $\alpha_C := (1/2)^{\lceil \log(1/\beta_C)\rceil}$. Then $|\DD'_1\cup \DD'_2\cup \DD'_3|\ge \max\{\alpha_{C} k,1\}$ and since $\alpha_C >\beta_C/2$, we get $h(d) > \frac{\epsilon N}{2^5\alpha_{C}g(\epsilon) k}$, for all $d\in \DD'_1\cup \DD'_2\cup \DD'_3$. In order to obtain the desired lower bounds on $\alpha_C$, note that if $\epsilon$ is tiny, we have $\beta_C\ge 1/k$, and hence $\alpha_C\ge 1/(2k)$; otherwise if $\epsilon$ is big, we have $\beta_C\ge \epsilon/8$, and thus $\alpha_C \ge \epsilon/16$.
\end{proof} 


We remark that till now we have not made any assumptions on the populations of the districts, i.e., Lemma \ref{lem:mov-lemma} works for the setting of arbitrary district populations. In the remainder of the section, we derive our structural results assuming all districts have equal population $N/k$. We then leverage these observations to design our algorithm for the equal population case in Section \ref{sec:alg-equal-population}. 


We call a district $d\in \DD$ to be $W$-\emph{rich} if $W$ wins in $d$, and no other candidate can be made the winner in $d$ by changing at most $\frac{\epsilon}{2^8\alpha_{C}g(\epsilon)}$-fraction of votes in $d$. Also, we call $d$ to be $C$-\emph{poor} if $C$ is either the winner in $d$, or can be made the winner in $d$ by changing at most $\frac{\epsilon}{2^8\alpha_{C}g(\epsilon)}$-fraction of votes in $d$. 

\begin{lemma}
\label{lem:structure}
    The following statements hold.
    \begin{enumerate}[label=(\roman*)]
        \item All districts in $\DD'_3$ are $W$-rich.
        \item There does not exist any $C$-poor district in $(\DD_W\cup \DD'_1)\setminus ((\DD_2\setminus \DD'_2)\cup \DD_3)$.
    \end{enumerate}
\end{lemma}
\begin{proof}
    \begin{enumerate}[label=(\roman*)]
        \item Recall that in the strategy $\SM$, the winner is changed from $W$ to some candidate in $\CC \setminus \{W,C\}$ in the districts of $\DD_3$. Also by Lemma \ref{lem:mov-lemma}, we have $h(d)>\frac{\epsilon N}{2^5\alpha_{C}g(\epsilon)k}$ for all $d\in \DD'_3$. It follows from the optimality of the strategy $\SM$ that changing the winner requires altering at least $\frac{\epsilon}{2^5\alpha_{C}g(\epsilon)}$-fraction of votes in any district of $\DD'_3$, and thus all such districts are $W$-rich.
        \item First note that since the winner was changed to $C$ in the districts of $\DD_1 \cup \DD_2$ in the strategy $\SM$ and $h(d)>\frac{\epsilon N}{2^5\alpha_{C}g(\epsilon)k}$ for all $d\in \DD'_1\cup \DD'_2$, it follows from the optimality of the strategy $\SM$ that no district in $\DD'_1 \cup \DD'_2$ is $C$-poor. Suppose there exists a $C$-poor district $d\in \DD_W\setminus (\DD_2\cup \DD_3)$. Let $d' \in \DD'_1 \cup \DD'_2\cup \DD'_3$. Consider an alternate strategy $\SM'$ that alters the same set of votes as $\SM$ in the districts of $(\DD_1\cup \DD_2\cup \DD_3) \setminus \{d'\}$ and changes the winner from $W$ to $C$ in the district $d$. Clearly $\SM'$ is $C$-feasible and alters strictly fewer votes than $\SM$, contradicting the optimality of $\SM$.
    \end{enumerate}
\end{proof}


Let $\DD'_W := \DD_W \setminus (\DD_2\cup (\DD_3\setminus \DD'_3))$ and $\DD'_C:= \DD_C \cup (\DD_1\setminus \DD'_1)\cup (\DD_2\setminus \DD'_2)$ (see Figure \ref{fig:district-decomposition}). We now show the following result about $\DD'_W$ and $\DD'_C$.

\begin{lemma}
\label{lem:district-structure}
    The following statements hold.
    \begin{enumerate}[label=(\roman*)]
        \item If $|\DD_1|+2|\DD_2|+|\DD_3|=\Delta_C+1$, then all districts in $\DD'_W$ are $W$-rich, and no district in $\DD \setminus \DD'_C$ is $C$-poor.        
        \item If $|\DD_1|+2|\DD_2|+|\DD_3|=\Delta_C$, then the sum of the numbers of districts in $\DD'_W$ that are not $W$-rich and districts in $\DD \setminus \DD'_C$ that are $C$-poor is at most $1$.
    \end{enumerate}
\end{lemma}
\begin{proof}
    \begin{enumerate}[label=(\roman*)]
        \item First, observe that in this case, $\DD_1\cup \DD_3 = \emptyset$, otherwise we can consider an alternate strategy $\SM'$ that alters the same set of votes as $\SM$ except in one district of $\DD_1\cup \DD_3$. $\SM'$ is $C$-feasible since in the resulting election, $C$ and $W$ win in equal number of districts, thus contradicting the optimality of $\SM$. Suppose there exists a district $d \in \DD'_W$ that is not $W$-rich. Then in $d$, the winner can be changed from $W$ to some other candidate by changing only $\frac{\epsilon N}{2^8\alpha_{C}g(\epsilon)k}$ votes. Let $d' \in \DD'_2$. By Lemma \ref{lem:mov-lemma}, at least $\frac{\epsilon N}{2^5\alpha_{C}g(\epsilon)k}$ votes were altered in $d'$ in the strategy $\SM$. Consider a strategy $\overline{\SM}$ that alters the same set of votes as $\SM$ in the districts of $\DD_2 \setminus \{d'\}$ and changes the winner from $W$ to some other candidate in the district $d$ (which requires altering only $\frac{\epsilon N}{2^8\alpha_{C}g(\epsilon)k}$ votes). Observe that $\overline{\SM}$ is $C$-feasible since in the resulting election, $C$ and $W$ win in equal number of districts. But then $\overline{\SM}$ alters strictly less number of votes than $\SM$, a contradiction. Similarly if there exists a $d \in \DD \setminus \DD'_C$ that is $C$-poor, then by Lemma \ref{lem:structure}(ii) since no district in $\DD'_2$ is $C$-poor, it must be the case that $d\in \DD\setminus (\DD_W \cup \DD_C)$. Again, we could consider the strategy $\overline{\SM}$ that alters the same set of votes as $\SM$ in the districts of $\DD_2\setminus \{d'\}$ and changes the winner to $C$ in the district $d$. Clearly $\overline{\SM}$ is $C$-feasible and alters strictly less votes than $\SM$, a contradiction.
        
        \item Suppose there exists a district $d_W \in \DD'_W$ that is not $W$-rich and a district $d_C \in \DD \setminus \DD'_C$ that is $C$-poor. By Lemma \ref{lem:structure}(i), we must have $d_W\in \DD_W\setminus (\DD_2\cup \DD_3)$ and by Lemma \ref{lem:structure}(ii), we have $d_C\in (\DD \setminus (\DD_W\cup \DD_C \cup \DD_1))\cup \DD_3$. We first show that $\DD'_1 \cup \DD'_3 = \emptyset$ in this case. Assume to the contrary that there exists $d'\in \DD'_1 \cup \DD'_3$. Then we can consider an alternate $C$-feasible strategy $\SM'$ that alters the same set of votes as $\SM$ in the districts of $(\DD_1\cup \DD_2\cup \DD_3) \setminus \{d'\}$ and changes the winner from $W$ to some other candidate in the district $d_W$ (which requires altering at most $\frac{\epsilon N}{2^8\alpha_{C}g(\epsilon)k}$ votes). Clearly $\SM'$ alters fewer votes than $\SM$, a contradiction.       
        Next consider any $d' \in \DD'_2$. If $d\in \DD\setminus (\DD_W \cup \DD_C \cup \DD_1)$, we let $\overline{\SM}$ be the strategy that alters the same set of votes as $\SM$ in the districts of $(\DD_1\cup \DD_2\cup \DD_3) \setminus \{d'\}$, changes the winner from $W$ to some other candidate in $d_W$, and changes the winner to $C$ in the district $d_C$. Otherwise if $d\in \DD_3$, let $\overline{\SM}$ be the strategy that alters the same set of votes as $\SM$ in the districts of $(\DD_1\cup \DD_2\cup \DD_3) \setminus \{d,d'\}$, changes the winner from $W$ to some other candidate in $d_W$, and changes the winner to $C$ in the district $d_C$.
        Since by Lemma \ref{lem:mov-lemma}, we have $h(d') > \frac{\epsilon N}{2^5\alpha_{C}g(\epsilon) k}$, and the total number of votes altered in the districts $d_W$ and $d_C$ is only at most $\frac{\epsilon N}{2^7\alpha_{C}g(\epsilon) k}$, $\overline{\SM}$ alters strictly fewer votes than $\SM$, contradicting the optimality of $\SM$.       
    \end{enumerate}
\end{proof}

Let $\DD_W^{\text{rich}}$ be the set of $W$-rich districts and $\DD_C^{\text{poor}}$ be the set of $C$-poor districts.



\begin{lemma}
\label{lem:difference-large}
    It holds that $|\DD_W^{\text{rich}}|-|\DD_C^{\text{poor}}|\ge \max\{\alpha_{C}k-1,0\}$.
\end{lemma}
\begin{proof}
    Let $x_W$ be the number of districts in $\DD'_W$ that are not $W$-rich, and $x_C$ be the number of districts in $\DD \setminus \DD'_C$ that are $C$-poor. We divide the proof into two cases.


    \textbf{Case 1:} $|\DD_1|+2|\DD_2|+|\DD_3|=\Delta_C+1$ - In this case, we have $x_W=x_C=0$ from Lemma \ref{lem:district-structure}(i). Therefore
    \[|\DD_W^{\text{rich}}| \ge |\DD'_W|= |\DD_W|-|\DD_2|-(|\DD_3|-|\DD'_3|).\]
    Similarly,
    \[ |\DD_C^{\text{poor}}| \le |\DD'_C|= |\DD_C|+(|\DD_1|-|\DD'_1|)+(|\DD_2|-|\DD'_2|).\]
    Therefore, 
    \[|\DD_W^{\text{rich}}|-|\DD_C^{\text{poor}}|\ge (|\DD_W|-|\DD_C|)-(|\DD_1|+2|\DD_2|+|\DD_3|)+(|\DD'_1|+|\DD'_2|+|\DD'_3|)\ge \max\{\alpha_{C}k-1,0\}.\]

    \textbf{Case 2:} $|\DD_1|+2|\DD_2|+|\DD_3|=\Delta_C$ - In this case, $x_W+x_C\le 1$ from Lemma \ref{lem:district-structure}(ii). Similar to the previous case, we have 
    \[|\DD_W^{\text{rich}}| \ge |\DD'_W|-x_W = |\DD_W|-|\DD_2|-(|\DD_3|-|\DD'_3|)-x_W,\]
    and
    \[ |\DD_C^{\text{poor}}| \le |\DD'_C|+x_C= |\DD_C|+(|\DD_1|-|\DD'_1|)+(|\DD_2|-|\DD'_2|)+x_C.\]
    Hence
    \[|\DD_W^{\text{rich}}|-|\DD_C^{\text{poor}}|\ge (|\DD_W|-|\DD_C|)-(|\DD_1|+2|\DD_2|+|\DD_3|)+(|\DD'_1|+|\DD'_2|+|\DD'_3|)-(x_W+x_C)\ge \max\{\alpha_{C}k-1,0\},\]
    and we are done.
\end{proof}


\section{Algorithm for equal population}
\label{sec:alg-equal-population}
As a consequence of Lemma \ref{lem:difference-large}, if we are able to predict $W$ as the winner in the districts of $\DD_W^{\text{rich}}$, and not predict $C$ as the winner in any district outside $\DD_C^{\text{poor}}$, we would still have a gap of $\max\{\alpha_{C}k-1,0\}$ districts between these sets. Also observe that $\widetilde{O}(\frac{\alpha_C^2}{\epsilon^2})$ samples inside a district suffice to decide whether it belongs to $\DD_W^{\text{rich}}$ or $\DD_C^{\text{poor}}$ (or neither). This is the main intuition behind our algorithm, which is summarised in Algorithm \ref{alg:WP_knowndist}.


\begin{algorithm}[H]
\caption{}
\label{alg:WP_knowndist}
\begin{algorithmic}[1]
\Require $\delta \in (0,1)$.

\State $\widetilde{\epsilon} \gets \frac{1}{2}$.
\While{$\widetilde{\epsilon} > \max\{\frac{2^{15}}{k}\log \frac{\log(1/\epstl)}{\delta\epstl},\frac{1}{\sqrt{N}}\}$}
\State $\widetilde{\alpha} \gets \frac{1}{2}$. 
 \While{$\widetilde{\alpha} \geq \widetilde{\epsilon}/16$} 
        \LongState{
            Sample a set $\widetilde{\DD}_{ \widetilde{\epsilon}, \widetilde{\alpha}}$ of 
             $q(\frac{\widetilde{\alpha}}{2^7}, \frac{\delta\epstl}{16m\log(1/\widetilde{\epsilon})})$ 
            districts uniformly at random and with replacement.
            }

        \LongState{
            From each district in $\widetilde{\DD}_{ \widetilde{\epsilon}, \widetilde{\alpha}}$, 
            sample and query $\ell_{\widetilde{\epsilon}, \widetilde{\alpha}}:=q(\frac{\widetilde{\epsilon}}{2^{11} \widetilde{\alpha}\log(1/\widetilde{\epsilon})}, \frac{\widetilde{\epsilon}}{2^{14}})$ 
            votes uniformly at random and with replacement.
            }\label{line:samples}
        
        \LongState{
            For each $C \in \CC$, let 
            $\widehat{\DD}_{\widetilde{\epsilon}, \widetilde{\alpha}, C}\subseteq
            \widetilde{\DD}_{\widetilde{\epsilon},  \widetilde{\alpha}}$ be the districts where 
            the fraction of sampled votes received by $C$ exceeds the fraction 
            received by any $C'\neq C$ by strictly more than $\frac{\widetilde{\varepsilon}}{2^9\widetilde{\alpha} \log(1/\widetilde{\epsilon})}$ and let $\veehat{\DD}_{\widetilde{\epsilon}, \widetilde{
            \alpha}, C} \subseteq
            \widetilde{\DD}_{\widetilde{\epsilon}, \widetilde{\alpha}}$ be the districts in which $C$ is either the winner or can be made the winner by altering at most $\frac{\widetilde{\varepsilon}}{2^9\widetilde{\alpha} \log(1/\widetilde{\epsilon})}$-fraction of sampled votes.}

        \LongState{
        For each $C\in \CC$, let 
        $\widehat{f}_{\widetilde{\epsilon}, \widetilde{\alpha}}(C):= \frac{|\widehat{\DD}_{\widetilde{\epsilon}, \widetilde{\alpha}, C}|}{|\DD_{\widetilde{\epsilon}, \widetilde{\alpha}}|}k$ and $\veehat{f}_{\widetilde{\epsilon}, \widetilde{\alpha}}(C):=  \frac{|\veehat{\DD}_{\widetilde{\epsilon}, \widetilde{\alpha}, C}|}{|\DD_{\widetilde{\epsilon}, \widetilde{\alpha}}|}k$
        }
     \State $\widetilde{\alpha} \gets \frac{\widetilde{\alpha}}{2}$. 
   \EndWhile

\algstore{algo-equal-pop}
\end{algorithmic}
\end{algorithm}
\begin{algorithm}[H]
\begin{algorithmic}[1]
\algrestore{algo-equal-pop}
   
    \If{$\exists C\in\CC$ such that $\forall C'\neq C, \exists \widetilde{\alpha}_{C'}$ such that
        $\widehat{f}_{\widetilde{\epsilon},\widetilde{\alpha}_{C'}}(C)-\veehat{f}_{\widetilde{\epsilon},\widetilde{\alpha}_{C'}}(C') \ge \frac{\widetilde{\alpha}_{C'}k}{2}$}
            \State \textbf{return} $C$ as the winner. \label{line:big-epsilon}
    \EndIf
    \State $\widetilde{\epsilon} \gets \frac{\widetilde{\epsilon}}{2}$.
\EndWhile

\While{$\widetilde{\epsilon} > \frac{1}{\sqrt{N}}$}
\State $\widetilde{\alpha} \gets \frac{1}{2}$. 
    \While{$\widetilde{\alpha} \ge \frac{1}{2k}$} 
        \If{$\widetilde{\alpha} > \frac{8}{k}$}

            \LongState{
                Sample a set $\widetilde{\DD}_{ \widetilde{\epsilon}, \widetilde{\alpha}}$ of 
                 $q(\frac{\widetilde{\alpha}}{2^7}, \frac{\delta\epstl}{16m\log (2k)})$ 
                districts uniformly at random and with replacement.
                }
        \Else
            \State $\widetilde{\DD}_{ \widetilde{\epsilon}, \widetilde{\alpha}} \gets \DD$.
        \EndIf
            
        \LongState{
            From each district in $\widetilde{\DD}_{ \widetilde{\epsilon}, \widetilde{\alpha}}$, 
            sample and query $\ell_{\widetilde{\epsilon}, \widetilde{\alpha}}:=q(\frac{\widetilde{\epsilon}}{2^{11} \widetilde{\alpha}\log (2k)}, \frac{\delta\epstl}{2k\log (2k)})$ 
            votes uniformly at random and with replacement.
            }    

        \LongState{
            For each $C \in \CC$, let 
            $\widehat{\DD}_{\widetilde{\epsilon}, \widetilde{\alpha}, C}\subseteq
            \widetilde{\DD}_{\widetilde{\epsilon},  \widetilde{\alpha}}$ be the districts where 
            the fraction of sampled votes received by $C$ exceeds the fraction 
            received by any $C'\neq C$ by strictly more than $\frac{\widetilde{\varepsilon}}{2^9\widetilde{\alpha} \log (2k)}$ and let $\veehat{\DD}_{\widetilde{\epsilon}, \widetilde{
            \alpha}, C} \subseteq
            \widetilde{\DD}_{\widetilde{\epsilon}, \widetilde{\alpha}}$ be the districts in which $C$ is either the winner or can be made the winner by altering at most $\frac{\widetilde{\varepsilon}}{2^9\widetilde{\alpha} \log (2k)}$-fraction of sampled votes.}

        \LongState{
        For each $C\in \CC$, let 
        $\widehat{f}_{\widetilde{\epsilon}, \widetilde{\alpha}}(C):= \frac{|\widehat{\DD}_{\widetilde{\epsilon}, \widetilde{\alpha}, C}|}{|\DD_{\widetilde{\epsilon}, \widetilde{\alpha}}|}k$ and $\veehat{f}_{\widetilde{\epsilon}, \widetilde{\alpha}}(C):=  \frac{|\veehat{\DD}_{\widetilde{\epsilon}, \widetilde{\alpha}, C}|}{|\DD_{\widetilde{\epsilon}, \widetilde{\alpha}}|}k$
        }
     \State $\widetilde{\alpha} \gets \frac{\widetilde{\alpha}}{2}$. 

    \EndWhile
    \If{$\exists C\in\CC$ such that $\forall C'\neq C, \exists \widetilde{\alpha}_{C'}$ such that
        $\widehat{f}_{\widetilde{\epsilon},\widetilde{\alpha}_{C'}}(C)-\veehat{f}_{\widetilde{\epsilon},\widetilde{\alpha}_{C'}}(C') \ge \frac{\widetilde{\alpha}_{C'}k}{2}$ if $\widetilde{\alpha}_{C'}> \frac{8}{k}$, and $\widehat{f}_{\widetilde{\epsilon},\widetilde{\alpha}_{C'}}(C)-\veehat{f}_{\widetilde{\epsilon},\widetilde{\alpha}_{C'}}(C') \ge 0$ otherwise} 
            \State \textbf{return} $C$ as the winner.\label{line:tiny-epsilon}
    \EndIf
\EndWhile

\State Query all votes and return the winner.\label{line:return}
\end{algorithmic}
\end{algorithm}

In order to prove the correctness of our algorithm, we need to establish that whenever the algorithm returns some candidate as the winner for a particular pair of values $(\widetilde{\epsilon},\widetilde{\alpha})$, that candidate must be $W$ with high probability. Further, to bound the query complexity, we shall show that when the estimate $\widetilde{\epsilon}$ reaches $\epsilon$, the algorithm halts with high probability.


To this end, we call a district $d\in \widetilde{D}_{\widetilde{\epsilon},\widetilde{\alpha}}$ to be \emph{good} if for every $C\in \CC$, the difference between the fractions of actual and sampled votes received by $C$ in $d$ is at most $\frac{\widetilde{\epsilon}}{2^{11} \widetilde{\alpha}g(\widetilde{\epsilon})}$, and \emph{bad} otherwise.

\begin{lemma}
\label{lem:good-districts}
    The following statements hold.
    \begin{enumerate}[label=(\roman*)]
        \item For each big $\widetilde{\epsilon}$, with probability at least $1-\frac{\delta\epstl}{2}$, it holds that for every $\widetilde{\alpha}$, at least $(1-\frac{\widetilde{\alpha}}{2^6})$-fraction of districts of $\widetilde{D}_{\widetilde{\epsilon},\widetilde{\alpha}}$ are good.
        \item For each tiny $\widetilde{\epsilon}$, with probability at least $1-\frac{\delta\epstl}{2}$, it holds that for every $\widetilde{\alpha}$, all districts of $\widetilde{D}_{\widetilde{\epsilon},\widetilde{\alpha}}$ are good.       
    \end{enumerate}
\end{lemma}
\begin{proof}
\begin{enumerate}[label=(\roman*)]
    \item We fix a value of $\widetilde{\alpha}$ and consider the collection of sampled districts and votes obtained at the end of line \ref{line:samples}. An alternate way of obtaining this collection is the following -- we imagine we have a sample of $\ell_{\widetilde{\epsilon},\widetilde{\alpha}}$ votes uniformly at random and with replacement from each district in $\DD$, and then we sample $q(\frac{\widetilde{\alpha}}{2^7},\frac{\delta\epstl}{16m\log(1/\widetilde{\epsilon})})$ districts uniformly at random and with replacement and observe the sampled votes in those districts. By Lemma \ref{lem:dbmain}, the probability that a district is good is at least $1-\frac{\widetilde{\epsilon}}{2^{14}}$, and thus the expected number of bad districts is at most $\frac{\widetilde{\epsilon}k}{2^{14}}$. Therefore, by Lemma \ref{lem:chernoff-large-dev}, the probability that the number of bad districts exceeds $\frac{\widetilde{\epsilon} k}{2^{11}}$ is at most $2^{-\widetilde{\epsilon}k/2^{11}}\le \frac{\delta\epstl}{4\log(16/\widetilde{\epsilon})}$, where the last inequality follows since $\widetilde{\epsilon}$ is big. Conditioned on this, by Lemma \ref{lem:dbmain}, with probability at least $1-\frac{\delta\epstl}{16m\log(1/\widetilde{\epsilon})}\ge 1-\frac{\delta\epstl}{4\log (16/\widetilde{\epsilon})}$, the fraction of bad districts in $\widetilde{D}_{\widetilde{\epsilon},\widetilde{\alpha}}$ is at most $\frac{\widetilde{\epsilon}}{2^{11}}+\frac{\widetilde{\alpha}}{2^7}\le \frac{\widetilde{\alpha}}{2^6}$, since $\widetilde{\alpha}\ge \widetilde{\epsilon}/16$. Therefore, for every $\widetilde{\alpha}$, with probability at least $1-\frac{\delta\epstl}{2\log(16/\widetilde{\epsilon})}$, it holds that at least $(1-\frac{\widetilde{\alpha}}{2^6})$-fraction of districts of $\widetilde{D}_{\widetilde{\epsilon},\widetilde{\alpha}}$ are good. The claim then follows by taking a union bound over all the $\log(16/\widetilde{\epsilon})$ values of $\widetilde{\alpha}$.
    \item Once again, we consider the alternate view of obtaining the collection of sampled districts and votes -- we sample $\ell_{\widetilde{\epsilon},\widetilde{\alpha}}$ votes uniformly at random and with replacement from each district in $\DD$; if $\widetilde{\alpha}>8/k$, we sample $q(\frac{\widetilde{\alpha}}{2^7}, \frac{\delta\epstl}{16m\log (2k)})$ districts uniformly at random and with replacement, and observe the votes sampled in that district, otherwise if $\widetilde{\alpha}\le 8/k$, we observe the sampled votes in all districts. In this case, the probability that a district is good is at least $1-\frac{\delta\epstl}{2k\log (2k)}$ due to Lemma \ref{lem:dbmain}, and thus by a union bound over all districts, we get that the probability that all districts are good is at least $1-\frac{\delta\epstl}{2\log (2k)}$. Since the number of possible values of $\widetilde{\alpha}$ is at most $\log (2k)$, a union bound gives us that with probability at least $1-\frac{\delta\epstl}{2}$, for each value of $\widetilde{\alpha}$, all districts in $\widetilde{D}_{\widetilde{\epsilon},\widetilde{\alpha}}$ are good.
\end{enumerate}     
\end{proof}

Next, for any $p>0$ and $C\in \CC$, we let $\DD_{\widetilde{\epsilon}, \widetilde{\alpha},C,p}^{\text{rich}} \subseteq \DD$ be the districts where the fraction of votes received by $C$ exceeds the fraction received by any $C'\neq C$ by strictly more than $\frac{\widetilde{\varepsilon}}{p\widetilde{\alpha} g(\widetilde{\epsilon})}$. Also, let $\DD_{\widetilde{\epsilon}, \widetilde{\alpha},C,p}^{\text{poor}} \subseteq \DD$ be the districts where $C$ is either the winner or can be made the winner by altering at most $\frac{\widetilde{\varepsilon}}{p\widetilde{\alpha} g(\widetilde{\epsilon})}$-fraction of votes.

\begin{lemma}
\label{lem:large-small-districts}
    For any $\widetilde{\epsilon}$, with probability at least $1-\frac{\delta\epstl}{2}$, both the following events hold for each $\widetilde{\alpha}$ and each candidate $C\in \CC$.
    \begin{enumerate}[label=(\roman*)]
        \item The fraction of districts in $\widetilde{\DD}_{\widetilde{\epsilon}, \widetilde{\alpha}}$ that belong to $\DD_{\widetilde{\epsilon}, \widetilde{\alpha},C,p}^{\text{rich}}$ lies in $(\frac{|\DD_{\widetilde{\epsilon}, \widetilde{\alpha},C,p}^{\text{rich}}|}{k}-\frac{\widetilde{\alpha}}{2^7},\frac{|\DD_{\widetilde{\epsilon}, \widetilde{\alpha},C,p}^{\text{rich}}|}{k}+\frac{\widetilde{\alpha}}{2^7})$.
        \item The fraction of districts in $\widetilde{\DD}_{\widetilde{\epsilon}, \widetilde{\alpha}}$ that belong to $\DD_{\widetilde{\epsilon}, \widetilde{\alpha},C,p}^{\text{poor}}$ lies in $(\frac{|\DD_{\widetilde{\epsilon}, \widetilde{\alpha},C,p}^{\text{poor}}|}{k}-\frac{\widetilde{\alpha}}{2^7},\frac{|\DD_{\widetilde{\epsilon}, \widetilde{\alpha},C,p}^{\text{poor}}|}{k}+\frac{\widetilde{\alpha}}{2^7})$. 
    \end{enumerate}
\end{lemma}
\begin{proof}
    We describe the proof for $\DD_{\widetilde{\epsilon}, \widetilde{\alpha},C,p}^{\text{rich}}$; the proof for $\DD_{\widetilde{\epsilon}, \widetilde{\alpha},C,p}^{\text{poor}}$ is analogous. For a fixed $\widetilde{\alpha}$ and candidate $C\in \CC$, by Lemma \ref{lem:dbmain}, the probability that the fraction of districts of $\DD_{\widetilde{\epsilon}, \widetilde{\alpha},C,p}^{\text{rich}}$ is preserved up to an error of $\frac{\widetilde{\alpha}}{2^7}$ is at least $1-\frac{\delta\epstl}{16mg(\widetilde{\epsilon})}$. The desired claim follows by taking a union bound over all candidates and all the at most $8g(\widetilde{\epsilon})$ possible values of $\widetilde{\alpha}$.
\end{proof}

For each $\epstl$, let $\mathcal{E}_{\epstl}$ be the event that both the events mentioned in Lemma \ref{lem:large-small-districts} and either of the events stated in Lemma \ref{lem:good-districts} depending on whether $\epstl$ is big or tiny, hold. Then from Lemmas \ref{lem:good-districts} and \ref{lem:large-small-districts}, we have $\text{Pr}(\mathcal{E}_{\epstl})\ge 1-\delta\epstl$. Define $\mathcal{E}:= \cap_{\epstl}\mathcal{E}_{\epstl}$. Taking a union bound over all possible values of $\epstl$, we get that $\text{Pr}(\mathcal{E}) \ge 1-\delta\sum_{i\in \mathbb{N}}\frac{1}{2^i}=1-\delta$.

Let $\epsilon^*$ be the value of $\epstl$ when Algorithm \ref{alg:WP_knowndist} halts. We are now ready to show that our algorithm returns the true winner $W$ with high probability.

\begin{lemma}
\label{lem:correct-when-halt}
    Conditioned on the event $\mathcal{E}$, Algorithm \ref{alg:WP_knowndist} returns $W$ as the winner. 
\end{lemma}
\begin{proof}
    Let $C^*$ be the candidate returned by Algorithm \ref{alg:WP_knowndist}.
    Let $C'\neq C^*$ be any other candidate. First consider the case when $\epsilon^*$ is big, and so Algorithm \ref{alg:WP_knowndist} halts at line \ref{line:big-epsilon}. Let $\widetilde{\DD}_{\epsilon^*,\widetilde{\alpha}_{C'}}^{\text{bad}} \subseteq \widetilde{\DD}_{\epsilon^*,\widetilde{\alpha}_{C'}}$ be the set of sampled districts that are bad, i.e., for some candidate $C$, the absolute difference between the fractions of actual and sampled votes received by $C$ exceeds $\frac{\widetilde{\epsilon}}{2^{11} \widetilde{\alpha}_{C'}\log(1/\widetilde{\epsilon})}$. Then from Lemma \ref{lem:good-districts}(i), $\frac{|\widetilde{\DD}_{\epsilon^*,\widetilde{\alpha}_{C'}}^{\text{bad}}|}{|\widetilde{\DD}_{\epsilon^*,\widetilde{\alpha}_{C'}}|}\le \frac{\widetilde{\alpha}_{C'}}{2^6}$. In all districts of $\widehat{\DD}_{\epsilon^*, \widetilde{\alpha}_{C'}, C^*}\setminus \widetilde{\DD}_{\epsilon^*,\widetilde{\alpha}_{C'}}^{\text{bad}}$, the true fraction of votes received by $C^*$ exceeds that received by any other candidate by strictly more than $\frac{\epsilon^*}{2^9\widetilde{\alpha}_{C'}\log(1/\epsilon^*)}-\frac{\epsilon^*}{2^{11}\widetilde{\alpha}_{C'}\log(1/\epsilon^*)}>  \frac{\varepsilon^*}{2^{10}\widetilde{\alpha}_{C'} \log(1/\epsilon^*)}$. Applying Lemma \ref{lem:large-small-districts}(i) for $p=2^{10}$, it follows that the true fraction of districts where the fraction of votes received by $C^*$ exceeds the fraction received by any other candidate by strictly more than $\frac{\widetilde{\varepsilon}}{2^{10}\widetilde{\alpha}_{C'} \log(1/\widetilde{\epsilon})}$, is at least $\frac{|\widehat{\DD}_{\epsilon^*, \widetilde{\alpha}_{C'}, C^*}\setminus \widetilde{\DD}_{\epsilon^*,\widetilde{\alpha}_{C'}}^{\text{bad}}|}{|\widetilde{\DD}_{\epsilon^*,\widetilde{\alpha}_{C'}}|}-\frac{\widehat{\alpha}_{C'}}{2^7}$. In particular, $C^*$ wins in all such districts. Therefore $|\DD_{C^*}|\ge \widehat{f}_{\epsilon^*, \widetilde{\alpha}_{C'}}(C^*)-\frac{\widehat{\alpha}_{C'}k}{2^6}-\frac{\widehat{\alpha}_{C'}k}{2^7}\ge \widehat{f}_{\epsilon^*, \widetilde{\alpha}_{C'}}(C^*) - \frac{\widehat{\alpha}_{C'}k}{2^5}$. 

    Next, by Lemma \ref{lem:good-districts}(i), we have that in all districts of $\widetilde{\DD}_{\epsilon^*,\widetilde{\alpha}_{C'}}\setminus (\veehat{\DD}_{\epsilon^*, \widetilde{
            \alpha}_{C'}, C'}\cup \widetilde{\DD}_{\epsilon^*,\widetilde{\alpha}_{C'}}^{\text{bad}})$, at least $\frac{\epsilon^*}{2^9\widetilde{\alpha}_{C'}\log(1/\epsilon^*)}-\frac{\epsilon^*}{2^{11}\widetilde{\alpha}_{C'}\log(1/\epsilon^*)}>  \frac{\epsilon^*}{2^{10}\widetilde{\alpha}_{C'}\log(1/\epsilon^*)}$-fraction of votes need to be changed in order to make $C'$ as the winner. Therefore, by Lemma \ref{lem:large-small-districts}(ii), the true fraction of districts where $C'$ is either the winner or can be made the winner by altering at most $\frac{\epsilon^*}{2^{10} \widetilde{\alpha}_{C'}\log(1/\epsilon^*)}$-fraction of votes, is at most $\frac{|\veehat{\DD}_{\epsilon^*, \widetilde{
            \alpha}_{C'}, C'}\cup \widetilde{\DD}_{\epsilon^*,\widetilde{\alpha}_{C'}}^{\text{bad}}|}{|\widetilde{\DD}_{\epsilon^*,\widetilde{\alpha}_{C'}}|}+\frac{\widehat{\alpha}_{C'}}{2^7}$. In particular, the set $\DD_{C'}$ of districts where $C'$ wins satisfies $|\DD_{C'}|\le \veehat{f}_{\epsilon^*, \widetilde{\alpha}_{C'}}(C') + \frac{\widehat{\alpha}_{C'}k}{2^6}+\frac{\widehat{\alpha}_{C'}k}{2^7} \le \veehat{f}_{\epsilon^*, \widetilde{\alpha}_{C'}}(C')+\frac{\widehat{\alpha}_{C'}k}{2^5}$. Since $\widehat{f}_{\epsilon^*, \widetilde{\alpha}_{C'}}(C^*) - \veehat{f}_{\epsilon^*, \widetilde{\alpha}_{C'}}(C') \ge \frac{\widetilde{\alpha}_{C'}k}{2}$, we have $|\DD_{C^*}|-|\DD_{C'}|\ge \frac{\widetilde{\alpha}_{C'}k}{2}-\frac{\widetilde{\alpha}_{C'}k}{2^4}>0$, and therefore it must be the case that $C^*=W$.

    Next, suppose that $\epsilon^*$ is tiny, and $\widetilde{\alpha}_{C'}> \frac{8}{k}$. In this case, by Lemma \ref{lem:good-districts}(ii), we have $\widetilde{\DD}_{\epsilon^*,\widetilde{\alpha}_{C'}}^{\text{bad}} = \emptyset$. By a similar argument as above (with $\log (1/\epsilon^*)$ replaced by $\log (2k)$), we have $|\DD_{C^*}|\ge \widehat{f}_{\epsilon^*, \widetilde{\alpha}_{C'}}(C^*) - \frac{\widehat{\alpha}_{C'}k}{2^7}$. A similar argument for $C'$ yields $|\DD_{C'}| \le \veehat{f}_{\epsilon^*, \widetilde{\alpha}_{C'}}(C') + \frac{\widehat{\alpha}_{C'}k}{2^7}$. Since $\widehat{f}_{\epsilon^*, \widetilde{\alpha}_{C'}}(C^*) - \veehat{f}_{\epsilon^*, \widetilde{\alpha}_{C'}}(C') \ge \frac{\widetilde{\alpha}_{C'}k}{2}$, we again have $f(C^*)>f(C')$, and therefore $C^*=W$.

    Finally, consider the case when $\epsilon^*$ is tiny and $\widetilde{\alpha}_{C'} \le \frac{8}{k}$. Note that, in this case we have $\widetilde{\DD}_{\epsilon^*,\widetilde{\alpha}_{C'}}=\DD$. Again, by a similar argument as above, we conclude that $C^*$ must win in all districts of $\widehat{\DD}_{\epsilon^*, \widetilde{\alpha}_{C'}, C^*}$, and $C'$ is not the winner in any district of $\DD\setminus \veehat{\DD}_{\epsilon^*, \widetilde{
            \alpha}_{C'}, C'}$. Therefore, we have $|\DD_{C^*}|\ge \widehat{f}_{\epsilon^*, \widetilde{\alpha}_{C'}}(C^*)$ and $|\DD_{C'}|\le \veehat{f}_{\epsilon^*, \widetilde{\alpha}_{C'}}(C')$. Since $\widehat{f}_{\epsilon^*, \widetilde{\alpha}_{C'}}(C^*) - \veehat{f}_{\epsilon^*, \widetilde{\alpha}_{C'}}(C') \ge 0$, this implies that $|\DD_{C^*}|\ge |\DD_{C'}|$. Since the choice of $C'$ was arbitrary, by our assumption that the election $E$ has an unambiguous winner, it must be the case that $C^*=W$, and we are done.    
\end{proof}


Next, in order to bound the query complexity of our algorithm, we show that in the iteration when $\widetilde{\epsilon}=\epsilon$, Algorithm \ref{alg:WP_knowndist} must return the true winner $W$ with high probability.

\begin{lemma}
\label{lem:halt-surely}
    Conditioned on the event $\mathcal{E}$, it holds that $\epsilon^*\ge \epsilon$.
\end{lemma}
\begin{proof}
    It suffices to show that conditioned on $\mathcal{E}$, in the iteration when $\epstl=\epsilon$, Algorithm \ref{alg:WP_knowndist} must return $W$ as the winner.
    Let $C\neq W$ be any other candidate. Let $\alpha_C$ be the value guaranteed by Lemma \ref{lem:mov-lemma}. Recall that by Lemma \ref{lem:difference-large}, we have $|\DD_{\epsilon,\alpha_C,W,2^8}^{\text{rich}}|-|\DD^{\text{poor}}_{\epsilon,\alpha_C,C,2^8}|\ge \max\{\alpha_C k-1,0\}$.
    
    Consider first the case when $\epsilon$ is big. Let $\widetilde{\DD}^{\text{rich}}_{\epsilon,\alpha_C,W,2^8} = \DD_{\epsilon,\alpha_C,W,2^8}^{\text{rich}} \cap \widetilde{\DD}_{\epsilon,\alpha_C}$. Applying Lemma \ref{lem:large-small-districts}(i) with $p=2^8$, we have $\frac{|\widetilde{\DD}^{\text{rich}}_{\epsilon,\alpha_C,W,2^8}|}{|\widetilde{\DD}_{\epsilon,\alpha_C}|} \ge \frac{|\DD_{\epsilon,\alpha_C,W,2^8}^{\text{rich}}|}{k}-\frac{\alpha_C}{2^7}$. Let $\widetilde{\DD}_{\epsilon,\alpha_C}^{\text{bad}}\subseteq \DD_{\epsilon,\alpha_C}$ be the set of sampled districts that are bad. 
    By Lemma \ref{lem:good-districts}(i), $\frac{|\widetilde{\DD}_{\epsilon,\alpha_C}^{\text{bad}}|}{|\DD_{\epsilon,\alpha_C}|}\le \frac{\alpha_C}{2^6}$. In each district of $\widetilde{\DD}^{\text{rich}}_{\epsilon,\alpha_C,W,2^8} \setminus \widetilde{\DD}_{\epsilon,\alpha_C}^{\text{bad}}$, the fraction of sampled votes received by $W$ will exceed the fraction received by any other candidate by strictly more than $\frac{\epsilon}{2^8\alpha_C\log(1/\epsilon)} - \frac{\epsilon}{2^{11}\alpha_C\log (1/\epsilon)}> \frac{\epsilon}{2^9\alpha_C\log (1/\epsilon)}$. Therefore, $\widehat{f}_{\epsilon,\alpha_{C}}(W) \ge \frac{|\widetilde{\DD}^{\text{rich}}_{\epsilon,\alpha_C,W,2^8} \setminus \widetilde{\DD}_{\epsilon,\alpha_C}^{\text{bad}}|}{|\widetilde{\DD}_{\epsilon,\alpha_C}|}k\ge |\DD_{\epsilon,\alpha_C,W,2^8}^{\text{rich}}|-\frac{\alpha_C k}{2^7}-\frac{\alpha_C k}{2^6}\ge |\DD_{\epsilon,\alpha_C,W,2^8}^{\text{rich}}|-\frac{\alpha_C k}{2^5}$.
    
     Similarly, let $\widetilde{\DD}_{\epsilon,\alpha_C,C,2^8}^{\text{poor}}=\DD^{\text{poor}}_{\epsilon,\alpha_C,C,2^8} \cap \widetilde{\DD}_{\epsilon,\alpha_C}$, so that $\frac{|\widetilde{\DD}^{\text{poor}}_{\epsilon,\alpha_C,C,2^8}|}{|\widetilde{\DD}_{\epsilon,\alpha_C}|} \le \frac{|\DD_{\epsilon,\alpha_C,C,2^8}^{\text{poor}}|}{k}+\frac{\alpha_C}{2^7}$. Then in each district of $\widetilde{\DD}_{\epsilon,\alpha_C}\setminus (\widetilde{\DD}^{\text{poor}}_{\epsilon,\alpha_C,C,2^8} \cup \widetilde{\DD}_{\epsilon,\alpha_C}^{\text{bad}})$, the fraction of sampled votes that need to be changed to make $C$ the winner is at least $\frac{\epsilon}{2^8\alpha_C\log (1/\epsilon)}-\frac{\epsilon}{2^{11}\alpha_C\log (1/\epsilon)} > \frac{\epsilon}{2^9\alpha_C\log (1/\epsilon)}$. Thus $\veehat{f}_{\epsilon,\alpha_{C}}(C) \le |\DD_{\epsilon,\alpha_C,C,2^8}^{\text{poor}}|+\frac{\alpha_C k}{2^7}+\frac{\alpha_C k}{2^6}\le |\DD_{\epsilon,\alpha_C,C,2^8}^{\text{poor}}|+\frac{\alpha_C k}{2^5}$. Hence $\widehat{f}_{\epsilon,\alpha_{C}}(W) - \veehat{f}_{\epsilon,\alpha_{C}}(C) \ge |\DD_{\epsilon,\alpha_C,W,2^8}^{\text{rich}}|-|\DD_{\epsilon,\alpha_C,C,2^8}^{\text{poor}}|-\frac{\alpha_{C}k}{2^4}\ge \frac{15\alpha_{C}k}{16}-1>\frac{\alpha_{C}k}{2}$, where the last inequality holds since we are in the case when $k>\frac{2^{15}}{\epsilon}\ge \frac{2^{11}}{\alpha_C}$. Thus, Algorithm \ref{alg:WP_knowndist} will return $W$ as the winner in line \ref{line:big-epsilon}.

     Next, consider the case when $\epsilon$ is small and $\alpha_C>\frac{8}{k}$. In this case, by Lemma \ref{lem:good-districts}(ii), we have $\widetilde{\DD}_{\epsilon,\alpha_C}^{\text{bad}}=\emptyset$. Therefore, an analogous argument gives $\widehat{f}_{\epsilon,\alpha_{C}}(W) \ge |\DD_{\epsilon,\alpha_C,W,2^8}^{\text{rich}}|-\frac{\alpha_C k}{2^7}$, and $\veehat{f}_{\epsilon,\alpha_{C}}(C) \le |\DD_{\epsilon,\alpha_C,C,2^8}^{\text{poor}}|+\frac{\alpha_C k}{2^7}$, and hence $\widehat{f}_{\epsilon,\alpha_{C}}(W) - \veehat{f}_{\epsilon,\alpha_{C}}(C) \ge |\DD_{\epsilon,\alpha_C,W,2^8}^{\text{rich}}| - |\DD_{\epsilon,\alpha_C,C,2^8}^{\text{poor}}| - \frac{\alpha_{C}k}{2^6}\ge \frac{63\alpha_{C}k}{64}-1>\frac{\alpha_{C}k}{2}$, where the final inequality holds by our assumption that $k>\frac{8}{\alpha_C}$. Therefore, the condition of line \ref{line:tiny-epsilon} is satisfied and $W$ will be returned as the winner.

     Finally, we consider the case when $\epsilon$ is small and $\alpha_C\le \frac{8}{k}$. For this case, $\widetilde{\DD}_{\epsilon,\alpha_C}=\DD$, and thus $\widetilde{\DD}^{\text{rich}}_{\epsilon,\alpha_C,W,2^8} = \DD_{\epsilon,\alpha_C,W,2^8}^{\text{rich}}$. Therefore $\widehat{f}_{\epsilon,\alpha_{C}}(W) \ge |\DD_{\epsilon,\alpha_C,W,2^8}^{\text{rich}}|$ and $\veehat{f}_{\epsilon,\alpha_{C}}(C) \le |\DD_{\epsilon,\alpha_C,C,2^8}^{\text{poor}}|$. By Lemma \ref{lem:difference-large}, we get $\widehat{f}_{\epsilon,\alpha_{C}}(W)-\veehat{f}_{\epsilon,\alpha_{C}}(C) \ge 0$, and hence the condition in line \ref{line:tiny-epsilon} is satisfied. Thus, Algorithm \ref{alg:WP_knowndist} returns $W$ as the winner.
\end{proof}

We are now ready to prove the final result.

\begin{lemma}
\label{lem:sample-complexity-bound}
    With probability at least $1-\delta$, Algorithm \ref{alg:WP_knowndist} returns $W$ as the winner and queries at most $\widetilde{O}(\frac{1}{\epsilon^2}\log\frac{m}{\delta}\log\frac{1}{\delta})$ votes.
\end{lemma}
\begin{proof}
    Again, it suffices to show that the mentioned claims hold conditioned on the event $\mathcal{E}$. From Lemma \ref{lem:correct-when-halt}, we know that Algorithm \ref{alg:WP_knowndist} returns $W$ as the winner. Also owing to Lemma \ref{lem:halt-surely},     
    it suffices to bound the query complexity till the iteration when $\widetilde{\epsilon}=\epsilon$. Let $p(\widetilde{\epsilon})$ be the number of queried votes in a fixed iteration of Algorithm \ref{alg:WP_knowndist} with parameter $\widetilde{\epsilon}$. Consider first the case when $\widetilde{\epsilon}$ is big. Then we have $|\widetilde{\DD}_{\widetilde{\epsilon},\widetilde{\alpha}}|=O(\frac{1}{\widetilde{\alpha}^2}\log\frac{m\log(1/\widetilde{\epsilon})}{\delta\epstl})=\widetilde{O}(\frac{1}{\widetilde{\alpha}^2}\log\frac{m}{\delta})$, and $\ell_{\widetilde{\epsilon},\widetilde{\alpha}}=O(\frac{\widetilde{\alpha}^2\log^2(1/\widetilde{\epsilon})}{\widetilde{\epsilon}^2}\log\frac{1}{\widetilde{\epsilon}})=\widetilde{O}(\frac{\widetilde{\alpha}^2}{\widetilde{\epsilon}^2})$. Thus for a fixed $\widetilde{\alpha}$, the number of queried votes is at most $|\widetilde{\DD}_{\widetilde{\epsilon},\widetilde{\alpha}}|\cdot\ell_{\widetilde{\epsilon},\widetilde{\alpha}}=\widetilde{O}(\frac{1}{\widetilde{\epsilon}^2}\log\frac{m}{\delta})$. Since there are only $\log\frac{16}{\widetilde{\epsilon}}$ values of $\widetilde{\alpha}$, we have $p(\widetilde{\epsilon})=\widetilde{O}(\frac{1}{\widetilde{\epsilon}^2}\log\frac{m}{\delta})$.

    Next, suppose $\widetilde{\epsilon}$ is tiny. Then $k\le \frac{2^{15}}{\widetilde{\epsilon}}\log\frac{\log(1/\widetilde{\epsilon})}{\delta\epstl}$, and so $\log k = \widetilde{O}(1)$. Thus we have $\ell_{\widetilde{\epsilon},\widetilde{\alpha}}=O(\frac{\widetilde{\alpha}^2\log^2k}{\widetilde{\epsilon}^2}\log\frac{k\log k}{\delta})=\widetilde{O}(\frac{\widetilde{\alpha}^2}{\widetilde{\epsilon}^2}\log\frac{1}{\delta})$. For each $\widetilde{\alpha}>8/k$, we have $|\widetilde{\DD}_{\widetilde{\epsilon},\widetilde{\alpha}}|=O(\frac{1}{\widetilde{\alpha}^2}\log\frac{m\log k}{\delta\epstl})=\widetilde{O}(\frac{1}{\widetilde{\alpha}^2}\log\frac{m}{\delta})$, and hence $|\widetilde{\DD}_{\widetilde{\epsilon},\widetilde{\alpha}}|\cdot\ell_{\widetilde{\epsilon},\widetilde{\alpha}}=\widetilde{O}(\frac{1}{\widetilde{\epsilon}^2}\log\frac{m}{\delta}\log\frac{1}{\delta})$. Again, for each $\widetilde{\alpha}\le 8/k$, we have $|\DD_{\widetilde{\epsilon},\widetilde{\alpha}}|=k=O(\frac{1}{\widetilde{\alpha}})$ and therefore $|\widetilde{\DD}_{\widetilde{\epsilon},\widetilde{\alpha}}|\cdot\ell_{\widetilde{\epsilon},\widetilde{\alpha}}=\widetilde{O}(\frac{1}{\widetilde{\epsilon}^2}\log\frac{1}{\delta})$. Summing over all the $\log (2k)=\widetilde{O}(1)$ values of $\widetilde{\alpha}$, we therefore obtain $p(\widetilde{\epsilon})=\widetilde{O}(\frac{1}{\widetilde{\epsilon}^2}\log\frac{m}{\delta}\log\frac{1}{\delta})$.

    Hence, in either case, we have $p(\widetilde{\epsilon})=\widetilde{O}(\frac{1}{\widetilde{\epsilon}^2}\log\frac{m}{\delta}\log\frac{1}{\delta})$. By a geometric sum argument, the total number of votes queried by our algorithm up to the iteration when $\widetilde{\epsilon}=\epsilon$ is bounded by $\widetilde{O}(\frac{1}{\epsilon^2}\log\frac{m}{\delta}\log\frac{1}{\delta})$, and we are done.
\end{proof}

\section{Structural results for arbitrary population}
\label{sec:structure-arbitrary}

We begin by noting that in the proof of Lemma \ref{lem:mov-lemma}, we never used the fact that all districts have equal population. Therefore, the guarantee of Lemma \ref{lem:mov-lemma} continues to hold even when the population of districts are arbitrary. The only difference is that in the equal population setting, $h(d)>\frac{\epsilon N}{2^5\alpha_{C}g(\epsilon)k}$ implies that at least $\frac{\epsilon}{2^5\alpha_{C}g(\epsilon)}$-fraction of votes were changed in the district $d$ in the optimal $C$-feasible strategy; however, such a claim does not hold now since the population of $d$ can be arbitrarily larger than $\frac{N}{k}$.

We now introduce some notations. For any $d\in \DD$, we shall let $n_d$ denote the population of district $d$. Let $\DD^{1/2}:= \{d\in \DD \mid n_d \in (0,\frac{N}{k}]\}$, and for any $r\ge 1$ that is a power of $2$, let $\DD^r := \{d\in \DD \mid n_d \in (\frac{rN}{k},\frac{2rN}{k}]\}$. Define $\DD_{\widetilde{\epsilon}}^{\text{pop}}:= \{d\in \DD \mid n_d >\frac{2^9N}{\widetilde{\epsilon}k}\}$, if $\widetilde{\epsilon}$ is big, and $\DD_{\widetilde{\epsilon}}^{\text{pop}} = \emptyset$ otherwise. Clearly $|\DD_{\widetilde{\epsilon}}^{\text{pop}}|\le \frac{\widetilde{\epsilon}k}{2^9}$.

Similar as in Section \ref{sec:structural-result}, let $W$ be the winner of $E$, and for some candidate $C\in \CC$, we consider the optimal $C$-feasible strategy that alters votes in three sets of districts $\DD_1,\DD_2,\DD_3$. Also let $\DD'_i \subseteq \DD_i$ be the districts guaranteed by Lemma \ref{lem:mov-lemma} so that $h(d)>\frac{\epsilon N}{2^5\alpha_{C}g(\epsilon)k}$, for all $d\in \DD'_1\cup \DD'_2\cup \DD'_3$.
We now generalise our notions of $W$-rich and $C$-poor districts as defined in Section \ref{sec:structural-result}, to the setting of arbitrary population.
For any $r\le \frac{2^8}{\epsilon}$ (so that $\DD^r\cap \DD_{\epsilon}^{\text{pop}}=\emptyset$), we call a district $d\in \DD^r$ to be $W$-\emph{rich} if $W$ wins in $d$, and no other candidate can be made the winner in $d$ by changing at most $\frac{\epsilon}{2^8\alpha_{C}rg(\epsilon)}$-fraction of votes in $d$.  Similarly we call $d$ to be $C$-\emph{poor} if $C$ is either the winner in $d$ or can be made the winner by changing at most $\frac{\epsilon}{2^8\alpha_{C}rg(\epsilon)}$-fraction of votes in $d$.
Let $\DD'_W := \DD_W \setminus (\DD_2\cup (\DD_3\setminus \DD'_3))$ and $\DD'_C:= \DD_C \cup (\DD_1\setminus \DD'_1)\cup (\DD_2\setminus \DD'_2)$ (see Figure \ref{fig:district-decomposition}).


The following lemma is an analogue of Lemma \ref{lem:district-structure} in the setting of arbitrary population. Essentially the same guarantees continue to hold, except in the districts of $\DD_{\epsilon}^{\text{pop}}$.

\begin{lemma}
\label{lem:district-structure-arbitrary}
    The following statements hold.
    \begin{enumerate}[label=(\roman*)]
        \item If $|\DD_1|+2|\DD_2|+|\DD_3|=\Delta_C+1$, then all districts in $\DD'_W\setminus \DD_{\epsilon}^{\text{pop}}$ are $W$-rich, and no district in $\DD \setminus (\DD'_C\cup \DD_{\epsilon}^{\text{pop}})$ is $C$-poor.
        \item If $|\DD_1|+2|\DD_2|+|\DD_3|=\Delta_C$, then the sum of the numbers of districts in $\DD'_W\setminus \DD_{\epsilon}^{\text{pop}}$ that are not $W$-rich and districts in $\DD \setminus (\DD'_C\cup \DD_{\epsilon}^{\text{pop}})$ that are $C$-poor is at most $1$.
    \end{enumerate}
\end{lemma}

Let $\DD_{r,W}^{\text{rich}}$ be the districts of $\DD^r$ that are $W$-rich, and $\DD_{r,C}^{\text{poor}}$ be the districts of $\DD^r$ that are $C$-poor. Let $\DD_{W}^{\text{rich}} = \cup_r \DD_{r,W}^{\text{rich}}$, and $\DD_{C}^{\text{poor}} \cup_r \DD_{r,C}^{\text{poor}}$.

\begin{lemma}
\label{lem:difference-large-arbitrary}
    If $\epsilon$ is big, it holds that $|\DD_W^{\text{rich}}|-|\DD_C^{\text{poor}}|\ge \frac{3\alpha_{C}k}{4}$; otherwise, for tiny $\epsilon$, it holds that $|\DD_W^{\text{rich}}|-|\DD_C^{\text{poor}}|\ge \max\{\alpha_{C}k-1,0\}$.
\end{lemma}
\begin{proof}
    Note that if $\epsilon$ is tiny, $\DD_{\epsilon}^{\text{pop}}=\emptyset$, and by the same proof as in Lemma \ref{lem:difference-large}, we obtain $|\DD_W^{\text{rich}}|-|\DD_C^{\text{poor}}|\ge \max\{\alpha_{C}k-1,0\}$. Therefore, assume that $\epsilon$ is big. Then by Lemma \ref{lem:mov-lemma}, it holds that $\alpha_C\ge \epsilon/16$. Therefore $|\DD_{\epsilon}^{\text{pop}}|\le \frac{\epsilon k}{2^9}\le \frac{\alpha_C k}{2^5}$.
    As in the proof of Lemma \ref{lem:difference-large}, we let $x_W$ denote the number of districts in $\DD'_W\setminus \DD_{\epsilon}^{\text{pop}}$ that are not $W$-rich, and $x_C$ denote the number of districts in $\DD \setminus (\DD'_C\cup \DD_{\epsilon}^{\text{pop}})$ that are $C$-poor. Then from Lemma \ref{lem:district-structure-arbitrary}, we have 
    \[|\DD_W^{\text{rich}}|\ge |\DD_W|-|\DD_2|-(|\DD_3|-|\DD'_3|)-x_W-|\DD_{\epsilon}^{\text{pop}}|,\]
    and 
    \[|\DD_C^{\text{poor}}|\le |\DD_C|+(|\DD_1|-|\DD'_1|)+(|\DD_2|-|\DD'_2|)+x_C+|\DD_{\epsilon}^{\text{pop}}|.\]
    Thus $|\DD_W^{\text{rich}}|-|\DD_C^{\text{poor}}|\ge (\alpha_{C}k-1)-\frac{\alpha_{C}k}{2^4}\ge \frac{3\alpha_{C}k}{4}$, where the last inequality follows since $k>\frac{2^{15}}{\epsilon}\ge \frac{2^{11}}{\alpha_C}$.
\end{proof}

With the above lemmas in hand, we are now ready to present our algorithm for the case when districts have arbitrary populations.

\section{Algorithm for arbitrary population}
\label{sec:alg-arbitrary-population}

Our algorithm is summarised in Algorithm \ref{alg:WP-arbitrary-pop}. Since our definitions of rich and poor districts now depend on $r$, the number of votes sampled inside the districts of $\DD^r$ is a function of $r$. Specifically, for a particular value of $\epstl$ and $\alphatl$, we now need to query a random sample of $\widetilde{\Theta}(\frac{\alphatl^2r^2}{\epstl^2})$ votes inside a district of $\DD^r$ in order to decide whether it is rich or poor for any candidate. Also, we can tolerate an error of $\widetilde{O}(\alphatl k)$ in estimating the fraction of districts of $\DD^r$ that are rich or poor for each candidate. Since $|\DD^r|\le k/r$, this translates to a fractional error of $\widetilde{O}(\widetilde{\alpha}r)$. This is the main intuition behind our algorithm, which is described in Algorithm \ref{alg:WP-arbitrary-pop}.

\begin{algorithm}[H]
\caption{}
\label{alg:WP-arbitrary-pop}
\begin{algorithmic}[1]
\Require $\delta \in (0,1)$.

\State $\widetilde{\epsilon} \gets \frac{1}{2}$.
\While{$\widetilde{\epsilon} > \max\{\frac{2^{15}}{k}\log \frac{\log(1/\epstl)}{\delta\epstl},\frac{1}{\sqrt{N}}\}$}
\State $\widetilde{\alpha} \gets \frac{1}{2}$. 
        
 \While{$\widetilde{\alpha} \geq \widetilde{\epsilon}/16$} \label{line:alpha-loop}
        \State $r\gets \frac{1}{2}$.
            
        \While{$r\le \min\{\frac{2^{12}\log(1/\widetilde{\epsilon})}{\widetilde{\alpha}},\frac{2^8}{\widetilde{\epsilon}}\}$} \label{line:r-loop}
            
            \LongState{
                Sample a set $\widetilde{\DD}_{ \widetilde{\epsilon}, \widetilde{\alpha},r}$ of 
                 $q(\frac{\widetilde{\alpha}r}{2^{11}\log (1/\widetilde{\epsilon})}, \frac{\delta\epstl}{2^8m\log^2(1/\widetilde{\epsilon})})$ 
                districts uniformly at random and with replacement from $\DD^r$.
                }

            \If{$|\DD^r|\in (0,\frac{2^{19}}{\widetilde{\epsilon}r}\log\frac{1}{\widetilde{\epsilon}}\log\frac{\log(1/\widetilde{\epsilon})}{\delta\epstl}]$}
                    \LongState{
                    From each district in $\widetilde{\DD}_{ \widetilde{\epsilon}, \widetilde{\alpha},r}$, 
                sample and query $\ell_{\widetilde{\epsilon}, \widetilde{\alpha},r}:=q(\frac{\widetilde{\epsilon}}{2^{11} \widetilde{\alpha}r\log(1/\widetilde{\epsilon})}, \frac{\delta\epstl}{2^7|\DD^r|\log^2(1/\widetilde{\epsilon})})$ 
                votes uniformly at random and with replacement.
                    }  \label{line:small-Dr}

            \Else                   

                    \LongState{
                From each district in $\widetilde{\DD}_{ \widetilde{\epsilon}, \widetilde{\alpha},r}$, 
                sample and query $\ell_{\widetilde{\epsilon}, \widetilde{\alpha},r}:=q(\frac{\widetilde{\epsilon}}{2^{11} \widetilde{\alpha}r\log(1/\widetilde{\epsilon})}, \frac{\widetilde{\epsilon}r}{2^{18}\log(1/\widetilde{\epsilon})})$ 
                votes uniformly at random and with replacement.
                } \label{line:large-Dr}
            \EndIf

\algstore{algo-arbit-pop-1}
\end{algorithmic}
\end{algorithm}
\begin{algorithm}[H]
\begin{algorithmic}[1]
\algrestore{algo-arbit-pop-1}

            \LongState{
                For each $C \in \CC$, let 
                $\widehat{\DD}_{\widetilde{\epsilon}, \widetilde{\alpha},r, C}\subseteq
                \widetilde{\DD}_{\widetilde{\epsilon},  \widetilde{\alpha},r}$ be the districts where 
                the fraction of sampled votes received by $C$ exceeds the fraction 
                received by any $C'\neq C$ by strictly more than $\frac{\widetilde{\varepsilon}}{2^9\widetilde{\alpha}r \log(1/\widetilde{\epsilon})}$ and let $\veehat{\DD}_{\widetilde{\epsilon}, \widetilde{
                \alpha},r, C} \subseteq
                \widetilde{\DD}_{\widetilde{\epsilon}, \widetilde{\alpha},r}$ be the districts in which $C$ is either the winner or can be made the winner by altering at most $\frac{\widetilde{\varepsilon}}{2^9\widetilde{\alpha}r \log(1/\widetilde{\epsilon})}$-fraction of sampled votes.}
            
            \State $r\gets 2r$.

        \EndWhile

        \LongState{
        For each $C\in \CC$, let 
        $\widehat{f}_{\widetilde{\epsilon}, \widetilde{\alpha}}(C):= \sum_r \frac{|\widehat{\DD}_{\widetilde{\epsilon}, \widetilde{\alpha},r, C}|}{|\DD_{\widetilde{\epsilon}, \widetilde{\alpha},r}|}|\DD^r|$ and $\veehat{f}_{\widetilde{\epsilon}, \widetilde{\alpha}}(C):= \sum_r \frac{|\veehat{\DD}_{\widetilde{\epsilon}, \widetilde{\alpha},r, C}|}{|\DD_{\widetilde{\epsilon}, \widetilde{\alpha},r}|}|\DD^r|$.
        }
     \State $\widetilde{\alpha} \gets \frac{\widetilde{\alpha}}{2}$. 
   \EndWhile
    \If{$\exists C\in\CC$ such that $\forall C'\neq C, \exists \widetilde{\alpha}_{C'}$ such that
        $\widehat{f}_{\widetilde{\epsilon},\widetilde{\alpha}_{C'}}(C)-\veehat{f}_{\widetilde{\epsilon},\widetilde{\alpha}_{C'}}(C') \ge \frac{\widetilde{\alpha}_{C'}k}{2}$}
            \State \textbf{return} $C$ as the winner. \label{line:big-epsilon-arbitrary}
    \EndIf
    \State $\widetilde{\epsilon} \gets \frac{\widetilde{\epsilon}}{2}$.
\EndWhile

\While{$\widetilde{\epsilon} > \frac{1}{\sqrt{N}}$}
\State $\widetilde{\alpha} \gets \frac{1}{2}$. 
    \While{$\widetilde{\alpha} > \frac{1}{2k}$} 
        \State $r \gets \frac{1}{2}$.
        \While{$r\le \min\{\frac{2^{12}\log(2k)}{\widetilde{\alpha}},k\}$} \label{line:r-loop-tiny}
            \If{$\widetilde{\alpha} > \frac{8}{k}$}
    
                \LongState{
                    Sample a set $\widetilde{\DD}_{ \widetilde{\epsilon}, \widetilde{\alpha},r}$ of 
                     $q(\frac{\widetilde{\alpha}r}{2^{11}\log(2k)}, \frac{\delta\epstl}{2^8m\log^2 (2k)})$ 
                    districts uniformly at random and with replacement.
                    }
            \Else
                \State $\widetilde{\DD}_{ \widetilde{\epsilon}, \widetilde{\alpha},r} \gets \DD^r$.
            \EndIf
            
            \LongState{
                From each district in $\widetilde{\DD}_{ \widetilde{\epsilon}, \widetilde{\alpha},r}$, 
                sample and query $\ell_{\widetilde{\epsilon}, \widetilde{\alpha},r}:=q(\frac{\widetilde{\epsilon}}{2^{11} \widetilde{\alpha}r\log (2k)}, \frac{\delta\epstl}{4k\log^2 (2k)})$ 
                votes uniformly at random and with replacement.
                }

            \LongState{
                For each $C \in \CC$, let 
                $\widehat{\DD}_{\widetilde{\epsilon}, \widetilde{\alpha},r, C}\subseteq
                \widetilde{\DD}_{\widetilde{\epsilon},  \widetilde{\alpha},r}$ be the districts where 
                the fraction of sampled votes received by $C$ exceeds the fraction 
                received by any $C'\neq C$ by strictly more than $\frac{\widetilde{\varepsilon}}{2^9\widetilde{\alpha}r \log (2k)}$ and let $\veehat{\DD}_{\widetilde{\epsilon}, \widetilde{
                \alpha},r, C} \subseteq
                \widetilde{\DD}_{\widetilde{\epsilon}, \widetilde{\alpha},r}$ be the districts in which $C$ is either the winner or can be made the winner by altering at most $\frac{\widetilde{\varepsilon}}{2^9\widetilde{\alpha}r \log (2k)}$-fraction of sampled votes.}

        
            \State $r\gets 2r$.    
        \EndWhile
        
        \LongState{
        For each $C\in \CC$, let 
        $\widehat{f}_{\widetilde{\epsilon}, \widetilde{\alpha}}(C):= \sum_r \frac{|\widehat{\DD}_{\widetilde{\epsilon}, \widetilde{\alpha},r, C}|}{|\DD_{\widetilde{\epsilon}, \widetilde{\alpha},r}|}|\DD^r|$ and $\veehat{f}_{\widetilde{\epsilon}, \widetilde{\alpha}}(C):= \sum_r \frac{|\veehat{\DD}_{\widetilde{\epsilon}, \widetilde{\alpha},r, C}|}{|\DD_{\widetilde{\epsilon}, \widetilde{\alpha},r}|}|\DD^r|$. 
        }
     \State $\widetilde{\alpha} \gets \frac{\widetilde{\alpha}}{2}$. 

    \EndWhile
    \If{$\exists C\in\CC$ such that $\forall C'\neq C, \exists \widetilde{\alpha}_{C'}$ such that
        $\widehat{f}_{\widetilde{\epsilon},\widetilde{\alpha}_{C'}}(C)-\veehat{f}_{\widetilde{\epsilon},\widetilde{\alpha}_{C'}}(C') \ge \frac{\widetilde{\alpha}_{C'}k}{2}$ if $\widetilde{\alpha}_{C'}> \frac{8}{k}$, and $\widehat{f}_{\widetilde{\epsilon},\widetilde{\alpha}_{C'}}(C)-\veehat{f}_{\widetilde{\epsilon},\widetilde{\alpha}_{C'}}(C') \ge 0$ otherwise} 
            \State \textbf{return} $C$ as the winner. \label{line:tiny-epsilon-arbitrary}
    \EndIf
\EndWhile

\State Query all votes and return the winner.
\end{algorithmic}
\end{algorithm}

We generalise our definitions of good and bad districts from Section \ref{sec:alg-equal-population}. We call a district $d\in \widetilde{D}_{\widetilde{\epsilon},\widetilde{\alpha},r}$ to be \emph{good} if for every $C\in \CC$, the difference between the fractions of actual and sampled votes received by $C$ in $d$ is at most $\frac{\widetilde{\epsilon}}{2^{11} \widetilde{\alpha}rg(\widetilde{\epsilon})}$; otherwise, $d$ is said to be \emph{bad}. The following lemma is an analogue of Lemma \ref{lem:good-districts}.

\begin{lemma}
\label{lem:good-districts-arbitrary}
    The following statements hold.
    \begin{enumerate}[label=(\roman*)]
        \item For each big $\widetilde{\epsilon}$, with probability at least $1-\frac{\delta\epstl}{2}$, it holds that for every $\widetilde{\alpha},r$, at least $(1-\frac{\widetilde{\alpha}r}{2^{10}\log(1/\widetilde{\epsilon})})$-fraction of districts of $\widetilde{D}_{\widetilde{\epsilon},\widetilde{\alpha},r}$ are good.
        \item For each tiny $\widetilde{\epsilon}$, with probability at least $1-\frac{\delta\epstl}{2}$, it holds that for every $\widetilde{\alpha},r$, all districts of $\widetilde{D}_{\widetilde{\epsilon},\widetilde{\alpha},r}$ are good.       
    \end{enumerate}
\end{lemma}
\begin{proof}
\begin{enumerate}[label=(\roman*)]
    \item The proof is analogous to the proof of Lemma \ref{lem:good-districts}. Once again we consider the alternate view of our algorithm wherein we first sample $\ell_{\widetilde{\epsilon},\widetilde{\alpha},r}$ votes from each district of $\DD^r$ and then sample $q(\frac{\widetilde{\alpha}r}{2^{11}\log (1/\widetilde{\epsilon})}, \frac{\delta\epstl}{2^8m\log^2(1/\widetilde{\epsilon})})$ districts from $\DD^r$. If $|\DD^r|\in (0,\frac{2^{19}}{\widetilde{\epsilon}r}\log\frac{1}{\widetilde{\epsilon}}\log\frac{\log(1/\widetilde{\epsilon})}{\delta\epstl}]$, by our choice of parameters, the probability that a district is good is at least $1-\frac{\delta\epstl}{2^7|\DD^r|\log^2(1/\widetilde{\epsilon})}$, and therefore all districts in $\DD^r$ are good with probability at least $1-\frac{\delta\epstl}{2^7\log^2(1/\widetilde{\epsilon})}$. Otherwise if $|\DD^r|>\frac{2^{19}}{\widetilde{\epsilon}r}\log\frac{1}{\widetilde{\epsilon}}\log\frac{\log(1/\widetilde{\epsilon})}{\delta\epstl}$, each district is good with probability at least $1-\frac{\widetilde{\epsilon}r}{2^{18}\log(1/\widetilde{\epsilon})}$, and thus the expected number of bad districts in $|\DD^r|$ is at most $\frac{\widetilde{\epsilon}r}{2^{18}\log(1/\widetilde{\epsilon})}|\DD^r|$. By Lemma \ref{lem:chernoff-large-dev}, the probability that the number of bad districts exceeds $\frac{\widetilde{\epsilon}r}{2^{15}\log(1/\widetilde{\epsilon})}|\DD^r|$ is at most $2^{-\frac{\widetilde{\epsilon}r}{2^{15}\log(1/\widetilde{\epsilon})}|\DD^r|}\le \frac{\delta\epstl}{2^8\log^2(1/\widetilde{\epsilon})}$. Conditioned on this, by Lemma \ref{lem:dbmain}, it holds with probability at least $1-\frac{\delta\epstl}{2^8m\log^2(1/\widetilde{\epsilon})}\ge 1-\frac{\delta\epstl}{2^8\log^2(1/\widetilde{\epsilon})}$ that the fraction of bad districts in $\widetilde{\DD}_{\widetilde{\epsilon},\widetilde{\alpha},r}$ is at most $\frac{\widetilde{\epsilon}r}{2^{15}\log(1/\widetilde{\epsilon})}+\frac{\widetilde{\alpha}r}{2^{11}\log(1/\widetilde{\epsilon})}\le \frac{\widetilde{\alpha}r}{2^{10}\log(1/\widetilde{\epsilon})}$, since $\widetilde{\alpha}\ge \widetilde{\epsilon}/16$. Therefore, with probability at least $1-\frac{\delta\epstl}{2^7\log^2(1/\widetilde{\epsilon})}$, for each $\widetilde{\alpha}$ and $r$, it holds that at least $1-\frac{\widetilde{\alpha}r}{2^{10}\log(1/\widetilde{\epsilon})}$-fraction of districts in $\widetilde{\DD}_{\widetilde{\epsilon},\widetilde{\alpha},r}$ are good. The claim then follows by taking a union bound over all the $\log\frac{2^{10}}{\widetilde{\epsilon}}$ values of $r$ and the $\log\frac{2^{4}}{\widetilde{\epsilon}}$ values of $\widetilde{\alpha}$.
    \item For each $\widetilde{\alpha},r$, the probability that a district is good is at least $1-\frac{\delta\epstl}{4\log^2(2k)}$. The claim follows by taking a union bound over all the at most $\log (2k)$ values of $\widetilde{\alpha}$ and at most $\log (4k)$ values of $r$.
\end{enumerate}    
\end{proof}

Let $p>0$ and $C\in \CC$. For any $\widetilde{\epsilon},\widetilde{\alpha},r$, let $\DD_{\widetilde{\epsilon}, \widetilde{\alpha},r,C,p}^{\text{rich}} \subseteq \DD^r$ be the districts where the fraction of votes received by $C$ exceeds the fraction received by any $C'\neq C$ by strictly more than $\frac{\widetilde{\varepsilon}}{p\widetilde{\alpha}r g(\widetilde{\epsilon})}$. Also, let $\DD_{\widetilde{\epsilon}, \widetilde{\alpha},r,C,p}^{\text{poor}} \subseteq \DD^r$ be the districts where $C$ is either the winner or can be made the winner by altering at most $\frac{\widetilde{\varepsilon}}{p\widetilde{\alpha}r g(\widetilde{\epsilon})}$-fraction of votes. Let $\DD_{\widetilde{\epsilon}, \widetilde{\alpha},C,p}^{\text{rich}}=\cup_r \DD_{\widetilde{\epsilon}, \widetilde{\alpha},r,C,p}^{\text{rich}}$, and $\DD_{\widetilde{\epsilon}, \widetilde{\alpha},C,p}^{\text{poor}} = \cup_r \DD_{\widetilde{\epsilon}, \widetilde{\alpha},r,C,p}^{\text{poor}}$. The following lemma is an analogue of Lemma \ref{lem:large-small-districts} in the setting of arbitrary population. Essentially, it states that the true fraction of districts of $\DD_{\widetilde{\epsilon}, \widetilde{\alpha},r,C,p}^{\text{rich}}$ (and $\DD_{\widetilde{\epsilon}, \widetilde{\alpha},C,p}^{\text{poor}}$) in $\DD^r$ is preserved in the sampled districts for each value of $r$.

\begin{lemma}
\label{lem:large-small-districts-arbitrary}
    For any $\widetilde{\epsilon}$, with probability at least $1-\frac{\delta\epstl}{2}$, both the following events hold for each $\widetilde{\alpha},r$ and each candidate $C\in \CC$.
    \begin{enumerate}
        \item The fraction of districts in $\widetilde{\DD}_{\widetilde{\epsilon}, \widetilde{\alpha},r}$ that belong to $\DD_{\widetilde{\epsilon}, \widetilde{\alpha},r,C,p}^{\text{rich}}$ lies in $(\frac{|\DD_{\widetilde{\epsilon}, \widetilde{\alpha},C,p,r}^{\text{rich}}|}{|\DD^r|}-\frac{\widetilde{\alpha}r}{2^{11}g(\widetilde{\epsilon})},\frac{|\DD_{\widetilde{\epsilon}, \widetilde{\alpha},C,p,r}^{\text{rich}}|}{|\DD^r|}+\frac{\widetilde{\alpha}r}{2^{11}g(\widetilde{\epsilon})})$.
        \item The fraction of districts in $\widetilde{\DD}_{\widetilde{\epsilon}, \widetilde{\alpha},r}$ that belong to $\DD_{\widetilde{\epsilon}, \widetilde{\alpha},r,C,p}^{\text{poor}}$ lies in $(\frac{|\DD_{\widetilde{\epsilon}, \widetilde{\alpha},r,C,p}^{\text{poor}}|}{|\DD^r|}-\frac{\widetilde{\alpha}r}{2^{11}g(\widetilde{\epsilon})},\frac{|\DD_{\widetilde{\epsilon}, \widetilde{\alpha},r,C,p}^{\text{poor}}|}{|\DD^r|}+\frac{\widetilde{\alpha}r}{2^{11}g(\widetilde{\epsilon})})$. 
    \end{enumerate}
\end{lemma}

Similar to Section \ref{sec:alg-equal-population}, for each $\epstl$, we let $\mathcal{E}_{\epstl}$ be the event that both the events mentioned in Lemma \ref{lem:large-small-districts-arbitrary} and either of the events stated in Lemma \ref{lem:good-districts-arbitrary} depending on whether $\epstl$ is big or tiny, hold. Then $\text{Pr}(\mathcal{E}_{\epstl})\ge 1-\delta\epstl$. Letting $\mathcal{E}= \cap_{\epstl}\mathcal{E}_{\epstl}$, we obtain $\text{Pr}(\mathcal{E}) \ge 1-\delta\sum_{i\in \mathbb{N}}\frac{1}{2^i}=1-\delta$.

Let $\epsilon^*$ be the value of $\epstl$ when Algorithm \ref{alg:WP-arbitrary-pop} halts.

\begin{lemma}
\label{lem:correct-when-halt-arbitrary}
    Conditioned on the event $\mathcal{E}$, Algorithm \ref{alg:WP-arbitrary-pop} returns $W$ as the winner.
\end{lemma}
\begin{proof}
    The proof is analogous to the proof of Lemma \ref{lem:correct-when-halt}. Let $C^*$ be the candidate returned by Algorithm \ref{alg:WP-arbitrary-pop}, and let $C'\neq C^*$ be any other candidate. Suppose that $\epsilon^*$ is big. Consider the iteration of the while loop in line \ref{line:alpha-loop} with $\widetilde{\alpha}=\widetilde{\alpha}_{C'}$. Let $R_1$ be the set of values of $r$ considered by the while loop in line \ref{line:r-loop}, and let $R_2$ be the set of remaining values of $r$, i.e., those that lie in the range $(\frac{2^{12}\log(1/\epsilon^*)}{\widetilde{\alpha}_{C'}},\frac{2^8}{\epsilon^*}]$ (if any). For each $r\in R_1$, let $\widetilde{\DD}_{\epsilon^*,\widetilde{\alpha}_{C'},r}^{\text{bad}} \subseteq \widetilde{\DD}_{\epsilon^*,\widetilde{\alpha}_{C'},r}$ be the collection of sampled districts that are bad. By Lemma \ref{lem:good-districts-arbitrary}(i), $\frac{|\widetilde{\DD}_{\epsilon^*,\widetilde{\alpha}_{C'},r}^{\text{bad}}|}{|\widetilde{\DD}_{\epsilon^*,\widetilde{\alpha}_{C'},r}|}\le \frac{\widetilde{\alpha}_{C'}r}{2^{10}\log(1/\epsilon^*)}$. In all districts of $\widehat{\DD}_{\epsilon^*, \widetilde{\alpha}_{C'},r, C^*}\setminus \widetilde{\DD}_{\epsilon^*,\widetilde{\alpha}_{C'},r}^{\text{bad}}$, the true fraction of votes received by $C^*$ exceeds that received by any other candidate by strictly more than $\frac{\epsilon^*}{2^9\widetilde{\alpha}_{C'}r\log(1/\epsilon^*)}-\frac{\epsilon^*}{2^{11}\widetilde{\alpha}_{C'}r\log(1/\epsilon^*)}>  \frac{\varepsilon^*}{2^{10}\widetilde{\alpha}_{C'}r \log(1/\epsilon^*)}$. Applying Lemma \ref{lem:large-small-districts-arbitrary}(i) for $p=2^{10}$, it follows that in the election $E$, the fraction of districts in $\DD^r$ where the fraction of votes received by $C^*$ exceeds the fraction received by any other candidate by strictly more than $\frac{\widetilde{\varepsilon}}{2^{10}\widetilde{\alpha}_{C'}r \log(1/\widetilde{\epsilon})}$, is at least $\frac{|\widehat{\DD}_{\epsilon^*, \widetilde{\alpha}_{C'},r, C^*}\setminus \widetilde{\DD}_{\epsilon^*,\widetilde{\alpha}_{C'},r}^{\text{bad}}|}{|\widetilde{\DD}_{\epsilon^*,\widetilde{\alpha}_{C'},r}|}-\frac{\widehat{\alpha}_{C'}r}{2^{11}\log(1/\epsilon^*)}> \frac{|\widehat{\DD}_{\epsilon^*, \widetilde{\alpha}_{C'},r, C^*}|}{|\widetilde{\DD}_{\epsilon^*,\widetilde{\alpha}_{C'},r}|}-\frac{\widehat{\alpha}_{C'}r}{2^{9}\log(1/\epsilon^*)}$. Summing over all $r\in R_1$, it follows that $|\DD_{C^*}|\ge \widehat{f}_{\epsilon^*, \widetilde{\alpha}_{C'}}(C^*) - \sum_{r\in R_1} (\frac{\widehat{\alpha}_{C'}r}{2^{9}\log(1/\epsilon^*)}|\DD^r|) 
    \ge \widehat{f}_{\epsilon^*, \widetilde{\alpha}_{C'}}(C^*) - \frac{\widehat{\alpha}_{C'}k}{2^5}$, where the final inequality follows since $|\DD^r|\le k/r$ and $|R_1|$ is trivially upper bounded by $\log(2^{10}/\epsilon^*)$.

    Similarly, by Lemma \ref{lem:good-districts-arbitrary}(i), for each $r\in R_1$, in all districts of $\widetilde{\DD}_{\epsilon^*,\widetilde{\alpha}_{C'},r}\setminus (\veehat{\DD}_{\epsilon^*, \widetilde{
            \alpha}_{C'},r, C'}\cup \widetilde{\DD}_{\epsilon^*,\widetilde{\alpha}_{C'},r}^{\text{bad}})$, at least $\frac{\epsilon^*}{2^9\widetilde{\alpha}_{C'}r\log(1/\epsilon^*)}-\frac{\epsilon^*}{2^{11}\widetilde{\alpha}_{C'}r\log(1/\epsilon^*)}>  \frac{\epsilon^*}{2^{10}\widetilde{\alpha}_{C'}r\log(1/\epsilon^*)}$-fraction of votes need to be changed in order to make $C'$ as the winner. Therefore, by Lemma \ref{lem:large-small-districts-arbitrary}(ii), the true fraction of districts in $\DD^r$ where $C'$ is either the winner or can be made the winner by altering at most $\frac{\epsilon^*}{2^{10} \widetilde{\alpha}_{C'}r\log(1/\epsilon^*)}$-fraction of votes, is at most $\frac{|\veehat{\DD}_{\epsilon^*, \widetilde{
            \alpha}_{C'},r, C'}\cup \widetilde{\DD}_{\epsilon^*,\widetilde{\alpha}_{C'},r}^{\text{bad}}|}{|\widetilde{\DD}_{\epsilon^*,\widetilde{\alpha}_{C'},r}|}+\frac{\widehat{\alpha}_{C'}r}{2^{11}\log(1/\epsilon^*)} \le \frac{|\veehat{\DD}_{\epsilon^*, \widetilde{
            \alpha}_{C'},r, C'}|}{|\widetilde{\DD}_{\epsilon^*,\widetilde{\alpha}_{C'},r}|}+\frac{\widetilde{\alpha}_{C'}r}{2^9\log(1/\epsilon^*)}$. Now, it is important to note that in the election $E$, $C'$ might also have won in the districts of $\cup_{r\in R_2} \DD^r$, and also in the districts of $\DD_{\epsilon^*}^{\text{pop}}$. Observe that for each $r\in R_2$, $|\DD^r|\le \frac{k}{r}\le \frac{\widetilde{\alpha}_{C'}k}{2^{12}\log(1/\epsilon^*)}$. Also $|\DD_{\epsilon^*}^{\text{pop}}| \le \frac{\epsilon^*k}{2^9}\le \frac{\widetilde{\alpha}_{C'}k}{2^5}$.             
            Therefore, the set $\DD_{C'}$ of districts where $C'$ wins satisfies $|\DD_{C'}|\le \veehat{f}_{\epsilon^*, \widetilde{\alpha}_{C'}}(C') + \sum_{r\in R_1} (\frac{\widehat{\alpha}_{C'}r}{2^{9}\log(1/\epsilon^*)}|\DD^r|) + \sum_{r\in R_2} |\DD^r|+ |\DD_{\epsilon^*}^{\text{pop}}| \le \veehat{f}_{\epsilon^*, \widetilde{\alpha}_{C'}}(C') + \frac{\widetilde{\alpha}_{C'}k}{2^5}+\frac{\widetilde{\alpha}_{C'}k}{2^8}+\frac{\widetilde{\alpha}_{C'}k}{2^5}\le \veehat{f}_{\epsilon^*, \widetilde{\alpha}_{C'}}(C') + \frac{\widetilde{\alpha}_{C'}k}{2^3}$. Since $\widehat{f}_{\epsilon^*, \widetilde{\alpha}_{C'}}(C^*) - \veehat{f}_{\epsilon^*, \widetilde{\alpha}_{C'}}(C') \ge \frac{\widetilde{\alpha}_{C'}k}{2}$, we have $|\DD_{C^*}|-|\DD_{C'}|\ge \frac{\widetilde{\alpha}_{C'}k}{2}-\frac{\widetilde{\alpha}_{C'}k}{2^5}-\frac{\widetilde{\alpha}_{C'}k}{2^3}>0$, and therefore it must be the case that $C^*=W$.

            Next, suppose that $\epsilon^*$ is tiny, and $\widetilde{\alpha}_{C'}> \frac{8}{k}$. Again, let $R'_1$ be the set of values of $r$ considered in the while loop in line \ref{line:r-loop-tiny}, and $R'_2$ be the set of remaining values of $r$. In this case, by Lemma \ref{lem:good-districts-arbitrary}(ii), we have $\widetilde{\DD}_{\epsilon^*,\widetilde{\alpha}_{C'},r}^{\text{bad}} = \emptyset$ for all $r\in R'_1$. A similar argument as before (with $\log (1/\epsilon^*)$ replaced by $\log (2k)$) gives $|\DD_{C^*}|\ge \widehat{f}_{\epsilon^*, \widetilde{\alpha}_{C'}}(C^*) - \sum_{r\in R_1} (\frac{\widehat{\alpha}_{C'}r}{2^{11}\log(2k)}|\DD^r|)\ge \frac{\widehat{\alpha}_{C'}k}{2^{10}}$. A similar argument for $C'$ yields $|\DD_{C'}| \le \veehat{f}_{\epsilon^*, \widetilde{\alpha}_{C'}}(C') + \frac{\widehat{\alpha}_{C'}k}{2^{10}}+\sum_{r\in R'_2} |\DD^r|\le \veehat{f}_{\epsilon^*, \widetilde{\alpha}_{C'}}(C') + \frac{\widehat{\alpha}_{C'}k}{2^{9}}$, since $|\DD^r|\le \frac{k}{r}\le \frac{\widetilde{\alpha}_{C'}k}{2^{12}\log(2k)}$ (recall that in this case $\DD_{\epsilon^*}^{\text{pop}}=\emptyset$). Since $\widehat{f}_{\epsilon^*, \widetilde{\alpha}_{C'}}(C^*) - \veehat{f}_{\epsilon^*, \widetilde{\alpha}_{C'}}(C') \ge \frac{\widetilde{\alpha}_{C'}k}{2}$, we again have $f(C^*)>f(C')$, and therefore $C^*=W$.

            Finally, consider the case when $\epsilon^*$ is tiny and $\widetilde{\alpha}_{C'} \le \frac{8}{k}$. In this case, $R'_2 = \emptyset$, since $\frac{2^{12}\log(2k)}{\widetilde{\alpha}_{C'}}>k$. We have $\widetilde{\DD}_{\epsilon^*,\widetilde{\alpha}_{C'},r}=\DD^r$ for all $r$. By an identical argument, $C^*$ must win in all districts of $\cup_r \widehat{\DD}_{\epsilon^*, \widetilde{\alpha}_{C'},r, C^*}$, and $C'$ is not the winner in any district of $\DD\setminus (\cup_r\veehat{\DD}_{\epsilon^*, \widetilde{
            \alpha}_{C'},r, C'})$. Therefore, we have $|\DD_{C^*}|\ge \widehat{f}_{\epsilon^*, \widetilde{\alpha}_{C'}}(C^*)$ and $|\DD_{C'}|\le \veehat{f}_{\epsilon^*, \widetilde{\alpha}_{C'}}(C')$. Since $\widehat{f}_{\epsilon^*, \widetilde{\alpha}_{C'}}(C^*) - \veehat{f}_{\epsilon^*, \widetilde{\alpha}_{C'}}(C') \ge 0$, this implies that $|\DD_{C^*}|\ge |\DD_{C'}|$, and by our assumption that the election $E$ has an unambiguous winner, it must be the case that $C^*=W$, and we are done. 
\end{proof}

\begin{lemma}
\label{lem:halt-surely-arbitrary}
    Conditioned on the event $\mathcal{E}$, it holds that $\epsilon^*\ge \epsilon$.
\end{lemma}
\begin{proof}
    We show that in the iteration with $\epstl=\epsilon$, Algorithm \ref{alg:WP-arbitrary-pop} must return $W$ as the winner assuming the event $\mathcal{E}$ holds.
    
    First assume that $\epsilon$ is big. Let $C\neq W$ be any other candidate, and $\alpha_C$ be the value guaranteed by Lemma \ref{lem:mov-lemma}. Consider the iteration of the while loop in line \ref{line:alpha-loop} with $\widetilde{\alpha}=\alpha_C$. Let $R_1$ be the set of values of $r$ considered in the while loop of line \ref{line:r-loop}, i.e., the values of $r$ that are at most $\min\{\frac{2^{12}\log(1/\widetilde{\epsilon})}{\alpha_C},\frac{2^8}{\widetilde{\epsilon}}\}$. Let $R_2$ be the (possibly) remaining values of $r$, i.e., those that lie in the range $(\frac{2^{12}\log(1/\widetilde{\epsilon})}{\alpha_C},\frac{2^8}{\widetilde{\epsilon}}]$.
    
    For each $r \in R_1$, let $\widetilde{\DD}^{\text{rich}}_{\epsilon,\alpha_C,r,W,2^8} = \DD_{\epsilon,\alpha_C,r,W,2^8}^{\text{rich}} \cap \widetilde{\DD}_{\epsilon,\alpha_C,r}$. Applying Lemma \ref{lem:large-small-districts-arbitrary}(i) with $p=2^8$, we have $\frac{|\widetilde{\DD}^{\text{rich}}_{\epsilon,\alpha_C,r,W,2^8}|}{|\widetilde{\DD}_{\epsilon,\alpha_C,r}|} \ge \frac{|\DD_{\epsilon,\alpha_C,r,W,2^8}^{\text{rich}}|}{|\DD^r|}-\frac{\alpha_{C}r}{2^{11}\log(1/\epsilon)}$. Let $\widetilde{\DD}_{\epsilon,\alpha_C,r}^{\text{bad}}\subseteq \widetilde{\DD}_{\epsilon,\alpha_C,r}$ be the set of sampled districts that are bad. 
    By Lemma \ref{lem:good-districts-arbitrary}(i), $\frac{|\widetilde{\DD}_{\epsilon,\alpha_C,r}^{\text{bad}}|}{|\DD_{\epsilon,\alpha_C,r}|}\le \frac{\alpha_{C}r}{2^{10}\log(1/\epsilon)}$. In each district of $\widetilde{\DD}^{\text{rich}}_{\epsilon,\alpha_C,r,W,2^8} \setminus \widetilde{\DD}_{\epsilon,\alpha_C,r}^{\text{bad}}$, the fraction of sampled votes received by $W$ will exceed the fraction received by any other candidate by strictly more than $\frac{\epsilon}{2^8\alpha_{C}r\log(1/\epsilon)} - \frac{\epsilon}{2^{11}\alpha_{C}r\log (1/\epsilon)}> \frac{\epsilon}{2^9\alpha_{C}r\log (1/\epsilon)}$. Therefore, \begin{align*}\widehat{f}_{\epsilon,\alpha_{C}}(W) &\ge \sum_{r\in R_1}\frac{|\widetilde{\DD}^{\text{rich}}_{\epsilon,\alpha_C,r,W,2^8} \setminus \widetilde{\DD}_{\epsilon,\alpha_C,r}^{\text{bad}}|}{|\widetilde{\DD}_{\epsilon,\alpha_C,r}|}|\DD^r| \\
    &\ge \sum_{r\in R_1}|\DD_{\epsilon,\alpha_C,r,W,2^8}^{\text{rich}}|-\sum_{r\in R_1} (\frac{\alpha_C r}{2^{11}\log(1/\epsilon)}+\frac{\alpha_C r}{2^{10}\log(1/\epsilon)})|\DD^r| \\
    &\ge \sum_{r\in R_1}|\DD_{\epsilon,\alpha_C,r,W,2^8}^{\text{rich}}|-\frac{\alpha_C k}{2^5},
    \end{align*} where the final inequality follows since $|R_1|\le \log\frac{2^{10}}{\epsilon}$. Next, observe that for any $r\in R_2$, we have $|\DD_{\epsilon,\alpha_C,r,W,2^8}^{\text{rich}}|\le |\DD^r|\le \frac{k}{r}\le \frac{\alpha_{C}k}{2^{12}\log(1/\epsilon)}$. Thus $\sum_{r\in R_2} |\DD_{\epsilon,\alpha_C,r,W,2^8}^{\text{rich}}|\le \frac{\alpha_{C}k}{2^8}$. Therefore, $\widehat{f}_{\epsilon,\alpha_{C}}(W) \ge \sum_{r\in R_1\cup R_2} |\DD_{\epsilon,\alpha_C,r,W,2^8}^{\text{rich}}| - \frac{\alpha_{C}k}{2^5}-\frac{\alpha_{C}k}{2^8}\ge |\DD_{\epsilon,\alpha_C,W,2^8}^{\text{rich}}|-\frac{\alpha_{C}k}{2^4}$.

    Similarly, for each $r\in R_1$, let $\widetilde{\DD}_{\epsilon,\alpha_C,r,C,2^8}^{\text{poor}}=\DD^{\text{poor}}_{\epsilon,\alpha_C,r,C,2^8} \cap \widetilde{\DD}_{\epsilon,\alpha_C,r}$ for each $r$, so that $\frac{|\widetilde{\DD}^{\text{poor}}_{\epsilon,\alpha_C,r,C,2^8}|}{|\widetilde{\DD}_{\epsilon,\alpha_C,r}|} \le \frac{|\DD_{\epsilon,\alpha_C,r,C,2^8}^{\text{poor}}|}{|\DD^r|}+\frac{\alpha_{C}r}{2^{11}\log(1/\epsilon)}$. Then in each district of $\widetilde{\DD}_{\epsilon,\alpha_C,r}\setminus (\widetilde{\DD}^{\text{poor}}_{\epsilon,\alpha_C,r,C,2^8} \cup \widetilde{\DD}_{\epsilon,\alpha_C,r}^{\text{bad}})$, the fraction of sampled votes that need to be changed to make $C$ the winner is at least $\frac{\epsilon}{2^8\alpha_{C}r\log (1/\epsilon)}-\frac{\epsilon}{2^{11}\alpha_{C}r\log (1/\epsilon)} > \frac{\epsilon}{2^9\alpha_{C}r\log (1/\epsilon)}$. Thus, $\veehat{f}_{\epsilon,\alpha_{C}}(C) \le \sum_{r\in R_1}| \DD_{\epsilon,\alpha_C,r,C,2^8}^{\text{poor}}|+\sum_{r\in R_1} (\frac{\alpha_C r}{2^{11}\log(1/\epsilon)}+\frac{\alpha_C r}{2^{10}\log(1/\epsilon)})|\DD^r|\le  |\DD_{\epsilon,\alpha_C,C,2^8}^{\text{poor}}|+\frac{\alpha_C k}{2^5}$. Hence $\widehat{f}_{\epsilon,\alpha_{C}}(W) - \veehat{f}_{\epsilon,\alpha_{C}}(C) \ge |\DD_{\epsilon,\alpha_C,W,2^8}^{\text{rich}}|-|\DD_{\epsilon,\alpha_C,C,2^8}^{\text{poor}}|-\frac{\alpha_{C}k}{2^3}\ge \frac{5\alpha_{C}k}{8}>\frac{\alpha_{C}k}{2}$, where the penultimate inequality follows from Lemma \ref{lem:difference-large-arbitrary}. Thus, Algorithm \ref{alg:WP_knowndist} will return $W$ as the winner in line \ref{line:tiny-epsilon-arbitrary}.

    Next, consider the case when $\epsilon$ is tiny and $\alpha_C>\frac{8}{k}$. Again, let $R'_1$ be the set of values of $r$ considered in the while loop in line \ref{line:r-loop-tiny}, and $R'_2$ be the remaining values of $r$. In this case, by Lemma \ref{lem:good-districts-arbitrary}(ii), we have $\widetilde{\DD}_{\epsilon,\alpha_C,r}^{\text{bad}}=\emptyset$ for all $r\in R_1$. Also for $r\in R_2$, $|\DD^r|\le \frac{k}{r}\le \frac{\alpha_{C}k}{2^{12}\log (2k)}$, and so $\sum_{r\in R_2} |\DD_{\epsilon,\alpha_C,r,W,2^8}^{\text{rich}}|\le \frac{\alpha_{C}k}{2^{11}}$.    
    Therefore, an analogous argument gives $\widehat{f}_{\epsilon,\alpha_{C}}(W) \ge \sum_{r\in R_1}|\DD_{\epsilon,\alpha_C,r,W,2^8}^{\text{rich}}|-\sum_{r\in R_1} (\frac{\alpha_C r}{2^{11}\log(2k)}|\DD^r|)\ge |\DD_{\epsilon,\alpha_C,W,2^8}^{\text{rich}}|-\frac{\alpha_{C}k}{2^{9}}$, and similarly $\veehat{f}_{\epsilon,\alpha_{C}}(C) \le |\DD_{\epsilon,\alpha_C,C,2^8}^{\text{poor}}|+\frac{\alpha_C k}{2^{10}}$, and hence $\widehat{f}_{\epsilon,\alpha_{C}}(W) - \veehat{f}_{\epsilon,\alpha_{C}}(C) \ge |\DD_{\epsilon,\alpha_C,W,2^8}^{\text{rich}}| - |\DD_{\epsilon,\alpha_C,C,2^8}^{\text{poor}}| - \frac{\alpha_{C}k}{2^8}\ge \frac{255\alpha_{C}k}{256}-1>\frac{\alpha_{C}k}{2}$, where the final inequality holds by our assumption that $k>\frac{8}{\alpha_C}$. Therefore, the condition of line \ref{line:tiny-epsilon-arbitrary} is satisfied, and $W$ will be returned as the winner.
    
    Finally, we consider the case when $\epsilon$ is small and $\alpha_C\le \frac{8}{k}$. For this case, $\widetilde{\DD}_{\epsilon,\alpha_C,r}=\DD^r$ for all $r$, and thus $\widetilde{\DD}^{\text{rich}}_{\epsilon,\alpha_C,r,W,2^8} = \DD_{\epsilon,\alpha_C,r,W,2^8}^{\text{rich}}$. Therefore $\widehat{f}_{\epsilon,\alpha_{C}}(W) \ge |\DD_{\epsilon,\alpha_C,r,W,2^8}^{\text{rich}}|$ and $\veehat{f}_{\epsilon,\alpha_{C}}(C) \le |\DD_{\epsilon,\alpha_C,r,C,2^8}^{\text{poor}}|$. By Lemma \ref{lem:difference-large-arbitrary}, we get $\widehat{f}_{\epsilon,\alpha_{C}}(W)-\veehat{f}_{\epsilon,\alpha_{C}}(C) \ge 0$, and hence the condition in line \ref{line:tiny-epsilon-arbitrary} is satisfied. Thus, Algorithm \ref{alg:WP-arbitrary-pop} returns $W$ as the winner.
\end{proof}

\begin{lemma}
\label{lem:sample-complexity-bound-arbitrary}
    With probability at least $1-\delta$, Algorithm \ref{alg:WP-arbitrary-pop} returns $W$ as the winner, and queries at most $\widetilde{O}(\frac{1}{\epsilon^2}\log\frac{m}{\delta}\log\frac{1}{\delta})$ votes.
\end{lemma}
\begin{proof}
    Due to Lemmas \ref{lem:correct-when-halt-arbitrary} and \ref{lem:halt-surely-arbitrary}, we obtain that Algorithm \ref{alg:WP-arbitrary-pop} must return $W$ as the winner with probability at least $1-\delta$. Similar to the proof of Lemma \ref{lem:sample-complexity-bound}, let $p(\widetilde{\epsilon})$ be the number of queried votes in the iteration of Algorithm \ref{alg:WP-arbitrary-pop} with parameter $\widetilde{\epsilon}$. It suffices to establish that $p(\widetilde{\epsilon})=\widetilde{O}(\frac{1}{\epsilon^2}\log\frac{m}{\delta}\log\frac{1}{\delta})$, from which the lemma would follow by a geometric sum argument over all values of $\widetilde{\epsilon}$ up to $\epsilon$.

    First, suppose that $\widetilde{\epsilon}$ is big. Then for each $\widetilde{\alpha}$ and $r$, we have $|\widetilde{\DD}_{\widetilde{\epsilon},\widetilde{\alpha},r}| = \widetilde{O}(\frac{1}{\widetilde{\alpha}^2r^2}\log\frac{m}{\delta})$. If $|\DD^r|\le \frac{2^{19}}{\widetilde{\epsilon}r}\log\frac{1}{\widetilde{\epsilon}}\log\frac{\log(1/\widetilde{\epsilon})}{\delta}$, then in line \ref{line:small-Dr}, we have $\ell_{\widetilde{\epsilon},\widetilde{\alpha},r}=\widetilde{O}(\frac{\widetilde{\alpha}^2r^2}{\widetilde{\epsilon}^2}\log\frac{|\DD^r|}{\delta})=\widetilde{O}(\frac{\widetilde{\alpha}^2r^2}{\widetilde{\epsilon}^2}\log\frac{1}{\delta})$, since $r\ge 1/2$. Otherwise, $\ell_{\widetilde{\epsilon},\widetilde{\alpha},r}=\widetilde{O}(\frac{\widetilde{\alpha}^2r^2}{\widetilde{\epsilon}^2})$ in line \ref{line:large-Dr}. Therefore, $|\widetilde{\DD}_{\widetilde{\epsilon},\widetilde{\alpha},r}|\cdot \ell_{\widetilde{\epsilon},\widetilde{\alpha},r}=\widetilde{O}(\frac{1}{\widetilde{\epsilon}^2}\log \frac{m}{\delta})$. Since the number of values of $\widetilde{\alpha}$ is $O(\log\frac{1}{\widetilde{\epsilon}})$, and for each $\widetilde{\alpha}$, there are $O(\log\frac{1}{\widetilde{\epsilon}})$ values of $r$, we obtain $p(\widetilde{\epsilon})=\widetilde{O}(\frac{1}{\widetilde{\epsilon}^2}\log \frac{m}{\delta})$.

    Next, when $\widetilde{\epsilon}$ is tiny, we have $\log k=\widetilde{O}(1)$. For any $\widetilde{\alpha},r$, it holds that $\ell_{\widetilde{\epsilon},\widetilde{\alpha},r}=\widetilde{O}(\frac{\widetilde{\alpha}^2r^2}{\widetilde{\epsilon}^2}\log\frac{1}{\delta})$. If $\widetilde{\alpha}>\frac{8}{k}$, we have $|\widetilde{\DD}_{\widetilde{\epsilon},\widetilde{\alpha},r}| = \widetilde{O}(\frac{1}{\widetilde{\alpha}^2r^2}\log\frac{m}{\delta})$, for all $r\le \min\{\frac{2^{12}\log(2k)}{\widetilde{\alpha}},k\}$. Thus, $|\widetilde{\DD}_{\widetilde{\epsilon},\widetilde{\alpha},r}|\cdot \ell_{\widetilde{\epsilon},\widetilde{\alpha},r}=\widetilde{O}(\frac{1}{\widetilde{\epsilon}^2}\log \frac{m}{\delta}\log\frac{1}{\delta})$. Since the number of values of $r$ is $\widetilde{O}(1)$, and the number of values of $\widetilde{\alpha}$ that exceed $8/k$ is also $\widetilde{O}(1)$, the total number of samples across all values of $\widetilde{\alpha}$ exceeding $8/k$ is bounded by $\widetilde{O}(\frac{1}{\widetilde{\epsilon}^2}\log \frac{m}{\delta}\log\frac{1}{\delta})$.     
    Finally if $\widetilde{\alpha}\le \frac{8}{k}$, we have $|\widetilde{\DD}_{\widetilde{\epsilon},\widetilde{\alpha},r}|=|\DD^r|\le \frac{k}{r}=O(\frac{1}{\widetilde{\alpha}r})$ for all $r$. Thus, $|\widetilde{\DD}_{\widetilde{\epsilon},\widetilde{\alpha},r}|\cdot \ell_{\widetilde{\epsilon},\widetilde{\alpha},r}=\widetilde{O}(\frac{\widetilde{\alpha}r}{\widetilde{\epsilon}^2}\log\frac{1}{\delta})$. Hence for each $\widetilde{\alpha}\le \frac{8}{k}$, the total number of samples is bounded by $\widetilde{O}(\frac{\widetilde{\alpha}}{\widetilde{\epsilon}^2}\log\frac{1}{\delta}) \cdot \sum_r r \le \widetilde{O}(\frac{\widetilde{\alpha}}{\widetilde{\epsilon}^2}\log\frac{1}{\delta}) \cdot k\le \widetilde{O}(\frac{1}{\widetilde{\epsilon}^2}\log\frac{1}{\delta})$. Since the number of such $\widetilde{\alpha}$ is only $O(1)$, we have $p(\widetilde{\epsilon})=\widetilde{O}(\frac{1}{\widetilde{\epsilon}^2}\log \frac{m}{\delta}\log\frac{1}{\delta})$ overall.  
    
    Thus, in either case, we have established that $p(\widetilde{\epsilon})=\widetilde{O}(\frac{1}{\widetilde{\epsilon}^2}\log \frac{m}{\delta}\log\frac{1}{\delta})$, completing the proof.
\end{proof}

\section{Conclusion and future research directions}
In this paper we propose a predictor for general district-based elections (with arbitrary district population and arbitrary number of candidates) with optimally bounded (up to polylog factors) query complexity for constant number of candidates, and has an additional logarithmic dependence on the number of candidates.

The immediate technical question that our work leaves open is to decide the necessity of the $\log m$ factor in the query complexity. Can it be removed or can a lower bound be proven?

Bhattacharyya and Dey studied the query complexity of winner prediction for various voting rules in the single-district setting. A natural direction of future research  is thus to study the query complexity of district-based elections for voting rules other than plurality. 

Another interesting set of questions is to estimate the margin of victory of the overall election (measure of the robustness of the election), as well as the margin of victory of the single-district election defined by the district winners (which corresponds to the robustness of the makeup of the House of Representatives in the US, the House of Commons in the UK or the Parliament in India).

\bibliographystyle{plain}
\bibliography{ref}

@inproceedings{KDS23,
  title={Sampling-Based Winner Prediction in District-Based Elections},
  author={Kar, Debajyoti and Dey, Palash and Sanyal, Swagato},
  booktitle={22nd International Conference on Autonomous Agents and Multiagent Systems},
  pages={2661--2663},
  year={2023},
  note={Full paper: arXiv:2203.00083}
}

@article{BD21,
  title={Predicting winner and estimating margin of victory in elections using sampling},
  author={Bhattacharyya, Arnab and Dey, Palash},
  journal={Artificial Intelligence},
  volume={296},
  pages={103476},
  year={2021},
  publisher={Elsevier}
}

@inproceedings{DM24,
  title={Evaluating District-based Election Surveys with Synthetic Dirichlet Likelihood},
  author={Mitra, Adway and Dey, Palash},
  booktitle={23rd International Conference on Autonomous Agents and Multiagent Systems},
  pages={1400--1408},
  year={2024}
}

@inproceedings{M24,
  title={Agent-based Simulation of District-based Elections with Heterogeneous Populations},
  author={Mitra, Adway},
  booktitle={22nd International Conference on Autonomous Agents and Multiagent Systems},
  pages={2730--2732},
  year={2023}
}

@article{BDW02,
  title={Complexity measures and decision tree complexity: a survey},
  author={Buhrman, Harry and De Wolf, Ronald},
  journal={Theoretical Computer Science},
  volume={288},
  number={1},
  pages={21--43},
  year={2002},
  publisher={Elsevier}
}

@InProceedings{AGJKLMSS18,
  author =	{Anshu, Anurag and Gavinsky, Dmitry and Jain, Rahul and Kundu, Srijita and Lee, Troy and Mukhopadhyay, Priyanka and Santha, Miklos and Sanyal, Swagato},
  title =	{A Composition Theorem for Randomized Query Complexity},
  booktitle =	{37th IARCS Annual Conference on Foundations of Software Technology and Theoretical Computer Science (FSTTCS 2017)},
  pages =	{10:1--10:13},
  ISBN =	{978-3-95977-055-2},
  ISSN =	{1868-8969},
  year =	{2018},
  publisher =	{Schloss Dagstuhl -- Leibniz-Zentrum f{\"u}r Informatik},
  doi =		{10.4230/LIPIcs.FSTTCS.2017.10}
}

@article{BDK18,
 author = {Ben-David, Shalev and Kothari, Robin},
 title = {Randomized Query Complexity of Sabotaged and Composed Functions},
 year = {2018},
 pages = {1--27},
 doi = {10.4086/toc.2018.v014a005},
 publisher = {Theory of Computing},
 journal = {Theory of Computing},
 volume = {14},
 number = {5}
}

@article{GLSS18,
 author = {Gavinsky, Dmytro and Lee, Troy and Santha, Miklos and Sanyal, Swagato},
 title = {Optimal Composition Theorem for Randomized Query Complexity},
 year = {2023},
 pages = {1--35},
 doi = {10.4086/toc.2023.v019a009},
 publisher = {Theory of Computing},
 journal = {Theory of Computing},
 volume = {19},
 number = {9}
}

@inproceedings{BDBGM22,
  title={Randomised composition and small-bias minimax},
  author={Ben-David, Shalev and Blais, Eric and G{\"o}{\"o}s, Mika and Maystre, Gilbert},
  booktitle={63rd Annual Symposium on Foundations of Computer Science (FOCS)},
  pages={624--635},
  year={2022},
  organization={IEEE}
}

@inproceedings{BDB20,
  title={A tight composition theorem for the randomized query complexity of partial functions},
  author={Ben-David, Shalev and Blais, Eric},
  booktitle={61st Annual Symposium on Foundations of Computer Science (FOCS)},
  pages={240--246},
  year={2020},
  organization={IEEE}
}

@inproceedings{CKMPS23,
  author =	{Chakraborty, Sourav and Kayal, Chandrima and Mittal, Rajat and Paraashar, Manaswi and Sanyal, Swagato and Saurabh, Nitin},
  title =	{On the Composition of Randomized Query Complexity and Approximate Degree},
  booktitle =	{Approximation, Randomization, and Combinatorial Optimization. Algorithms and Techniques (APPROX/RANDOM 2023)},
  pages =	{63:1--63:23},
  ISBN =	{978-3-95977-296-9},
  ISSN =	{1868-8969},
  year =	{2023},
  doi =	{10.4230/LIPIcs.APPROX/RANDOM.2023.63},
  publisher =	{Schloss Dagstuhl -- Leibniz-Zentrum f{\"u}r Informatik}
}

@inproceedings{S24,
  title={Randomized Query Composition and Product Distributions},
  author={Sanyal, Swagato},
  booktitle={41st International Symposium on Theoretical Aspects of Computer Science (STACS 2024)},
  pages={56--1},
  year={2024},
  organization={Schloss Dagstuhl--Leibniz-Zentrum f{\"u}r Informatik}
}

@inproceedings{GJ16,
  title={A composition theorem for conical juntas},
  author={G{\"o}{\"o}s, Mika and Jayram, TS},
  booktitle={31st Conference on Computational Complexity (CCC 2016)},
  pages={5--1},
  year={2016},
  organization={Schloss Dagstuhl--Leibniz-Zentrum f{\"u}r Informatik}
}

@article{MNSSTX16,
  title={Improved bounds for the randomized decision tree complexity of recursive majority},
  author={Magniez, Fr{\'e}d{\'e}ric and Nayak, Ashwin and Santha, Miklos and Sherman, Jonah and Tardos, G{\'a}bor and Xiao, David},
  journal={Random Structures \& Algorithms},
  volume={48},
  number={3},
  pages={612--638},
  year={2016},
  publisher={Wiley Online Library}
}

@misc{ADBD26,
      title={Monte Carlo to Las Vegas for Recursively Composed Functions}, 
      author={Bandar Al-Dhalaan and Shalev Ben-David},
      year={2026},
      note={Full paper: arXiv:2601.08073}
}

@INPROCEEDINGS{SW86,
  author={Saks, Michael and Wigderson, Avi},
  booktitle={27th Annual Symposium on Foundations of Computer Science (sfcs 1986)}, 
  title={Probabilistic Boolean decision trees and the complexity of evaluating game trees}, 
  year={1986},
  volume={},
  number={},
  pages={29-38},
  doi={10.1109/SFCS.1986.44}
}

@article{S91,
  title={On the Monte carlo boolean decision tree complexity of read-once formulae},
  author={Santha, Miklos},
  journal={Random Structures \& Algorithms},
  volume={6},
  number={1},
  pages={75--87},
  year={1995},
  publisher={Wiley Online Library}
}

@inproceedings{JK10,
author = {Jain, Rahul and Klauck, Hartmut},
title = {The Partition Bound for Classical Communication Complexity and Query Complexity},
year = {2010},
isbn = {9780769540603},
publisher = {IEEE Computer Society},
doi = {10.1109/CCC.2010.31},
booktitle = {25th Annual Conference on Computational Complexity (CCC 2010)},
pages = {247–258},
numpages = {12}
}

@article{GJPW18,
  title={Randomized communication versus partition number},
  author={G{\"o}{\"o}s, Mika and Jayram, TS and Pitassi, Toniann and Watson, Thomas},
  journal={ACM Transactions on Computation Theory (TOCT)},
  volume={10},
  number={1},
  pages={1--20},
  year={2018},
  publisher={ACM New York, NY, USA}
}

@inproceedings{AKK16,
author = {Ambainis, Andris and Kokainis, Martins and Kothari, Robin},
title = {Nearly optimal separations between communication (or query) complexity and partitions},
year = {2016},
isbn = {9783959770088},
publisher =	{Schloss Dagstuhl -- Leibniz-Zentrum f{\"u}r Informatik},
booktitle = {31st Conference on Computational Complexity (CCC 2016)},
articleno = {4},
numpages = {14}
}

@article{G98,
  title={Property testing and its connection to learning and approximation},
  author={Goldreich, Oded and Goldwasser, Shari and Ron, Dana},
  journal={Journal of the ACM (JACM)},
  volume={45},
  number={4},
  pages={653--750},
  year={1998},
  publisher={ACM New York, NY, USA}
}

@book{G17,
  title={Introduction to property testing},
  author={Goldreich, Oded},
  year={2017},
  publisher={Cambridge University Press}
}

@book{DP09,
  title={Concentration of measure for the analysis of randomized algorithms},
  author={Dubhashi, Devdatt P and Panconesi, Alessandro},
  year={2009},
  publisher={Cambridge University Press}
}

\begin{appendices}
\section{Proof of Lemma~\ref{lem:simple-alpha}}
\label{sec:simple-alpha-proof}
In this section, we present a proof of Lemma~\ref{lem:simple-alpha}, restated below.
\begin{lemma}[Restatement of Lemma~\ref{lem:simple-alpha}]
\label{lem:simple-alpha-restated}
Let $E$ be a district-based election with two candidates satisfying $\MOV(E) \geq \epsilon N$ and $4/k\leq \epsilon\leq 1/2$. Then there exists an $\alpha \in [\frac{\epsilon}{8}, \frac{1}{2}]$ such that there are at least $k(\frac{1}{2}+ \alpha)$ districts each with $\MOV$ at least $\widetilde{\Omega}(\frac{\epsilon N}{\alpha k})$ in which the overall winner wins.
\end{lemma}
\begin{proof}
Let $A$ and $B$ be the two candidates, and assume that $A$ is the winner of $E$. Let $\DD_A$ and $\DD_B$ be the sets of districts in which $A$ and $B$ win respectively. Let the districts in $\DD_A$ be, in non-decreasing order of their $\MOV$, $d_1, \ldots, d_{|\DD_w|}$. Define $\Delta:=|\DD_A| - |\DD_B|$. 


Consider the set $S:=\{d_1, \ldots, d_{\lceil\frac{\Delta}{2}\rceil\}}$ of $\lceil\frac{\Delta}{2}\rceil$ districts of least $\MOV$ in which $A$ wins. The next claim shows that $S$ contains a large number of districts with high $\MOV$.
\begin{claim}
\label{clm:margin-key}
There exists $\alpha' \in [\frac{\epsilon}{4}, 1]$ such that there are at least $\alpha' k$ districts in $S$ each of whose $\MOV$ is strictly more than $\frac{\epsilon N}{100\log(1/\epsilon)\alpha' k}$.
\end{claim}
Let us first finish the proof of the Lemma assuming the claim, and then prove the claim. $|\DD_A|=\frac{k}{2}+\frac{\Delta}{2}$. By Claim~\ref{clm:margin-key} there are at least $\epsilon k/4 \geq 1$ (by the assumptions of the Lemma statement) districts in $S$ with $\MOV$ at least $\frac{\epsilon N}{100\log(1/\epsilon)\alpha' k}$. By the definition of $S$, all districts in $\DD_A \setminus S$ have $\MOV$ at least $\frac{\epsilon N}{100\log(1/\epsilon)\alpha' k}$. Thus the total number of districts where $A$ wins and where the $\MOV$ is at least $\frac{\epsilon N}{100\log(1/\epsilon)\alpha' k}$ is at least $|\DD_A \setminus S|+\alpha' k \geq \frac{k}{2}+ \alpha' k -\frac{1}{2} \geq \frac{k}{2}+ (\alpha'/2) k$ (as $\alpha'k \geq 1$). The lemma follows by taking $\alpha=\alpha'/2$
\begin{proof}
For convenience, for $j=1,\ldots, \lceil\frac{\Delta}{2}\rceil$ we denote the district $d_j$ as $d'_{\lceil\frac{\Delta}{2}\rceil-j+1}$ so that $d'_1,\ldots, d'_{\lceil\frac{\Delta}{2}\rceil}$ is the sequence of districts in $S$ in non-increasing order of their $\MOV$. Define $\MOV(d'_j):=v_j$. We will show that there exists an $i \in [\lfloor \epsilon k/8\rfloor, \lceil\Delta/2\rceil]$ such that $v_i > \frac{\epsilon N}{100\log(1/\epsilon)i}$. This will imply the claim by taking $\alpha' = i/k$ and observing that $v_1\geq\ldots\geq v_i$.

Towards a contradiction, assume that for all $i \in [\lfloor \epsilon k/4\rfloor, \lceil\Delta/2\rceil]$, $N_i \leq \frac{\epsilon N}{100\log(1/\epsilon)i}$. We have that
\begin{align*}
\sum_{i=\lfloor \epsilon k/4\rfloor}^{\lceil\Delta/2\rceil} N_i &\leq \sum_{i=\lfloor \epsilon k/4\rfloor}^{\lceil\Delta/2\rceil}\frac{\epsilon N}{100\log(1/\epsilon)i} \\
&=\frac{\epsilon N}{100\log(1/\epsilon)}\sum_{i=\lfloor \epsilon k/4\rfloor}^{\lceil\Delta/2\rceil} \frac{1}{i} \\
&\leq\frac{\epsilon N}{100\log(1/\epsilon)}(H_k - H_{\epsilon k/4}) \tag*{$\lceil\Delta/2\rceil \leq k$} \\
&\leq \frac{\epsilon N}{100\log(1/\epsilon)} (\ln(4/\epsilon)+2) \tag*{Fact~\ref{fact:hn}} \\
&<\frac{\epsilon N}{5} \tag*{$\epsilon \leq 1/2$}.
\end{align*}
Now consider the election $E'$ obtained by altering votes in districts $D_{\lfloor\epsilon k/4\rfloor}, \ldots, D_{\lceil\Delta/2\rceil}$ such that $A$ is no more the winner in them. This can be accomplished by changing at most $\sum_{i=\lfloor \epsilon k/8\rfloor}^{\lceil\Delta/2\rceil} N_i < \frac{\epsilon N}{5} < \MOV(E)$ votes in $E$. Hence $A$ is the winner of $E'$ and furthermore $\MOV(E')> \MOV(E)-\frac{\epsilon N}{5} > \frac{4 \epsilon N}{5}$ votes. However, in the resulting election, the difference between the numbers of districts won by $A$ and $B$ is at most $\frac{\epsilon k}{8}-1$, which contradicts Lemma~\ref{lem:district-lebel-lb}.
\end{proof}
\end{proof}
\section{Additional preliminaries}
In this section, we include two results that we use in our proofs.
\begin{definition}
For positive integer $n$ the $n$-th harmonic sum $H_n$ is defined as
\[H_n=1+\frac{1}{2}+\ldots+\frac{1}{n}.\]
\end{definition}
\begin{fact}
\label{fact:hn}
For all $n\geq 1$, $\ln n \leq H_n \leq \ln n + 1$.
\end{fact}

We will need the following form of the Chernoff bound (see \cite{DP09} for a proof).
\begin{lemma}
\label{lem:chernoff-large-dev}
    Let $\{X_i\}_{i\in [n]}$ be $n$ independent $0-1$ random variables and $S=\sum_{i\in [n]} X_i$. Let $\mathbb{E}[S]=\mu$. Then for any $R\ge 2e\mu$, $\text{Pr}[S\ge R]\le 2^{-R}$.
\end{lemma}
\end{appendices}

\end{document}